\documentclass[sigconf]{acmart}

\acmConference[]{}{}{}
\renewcommand\footnotetextcopyrightpermission[1]{} 
\AtBeginDocument{%
  }

\usepackage{tikz,pgffor}
\usepackage{ifthen}
\usetikzlibrary{arrows,backgrounds,calc,patterns.meta}
\usetikzlibrary{automata,positioning,arrows,shapes,math,arrows.meta,decorations.pathmorphing}
\tikzstyle{min}=[thick,circle,draw,minimum size=1.4em,inner sep=0em,text centered]
\tikzstyle{dec}=[circle,draw,fill,minimum size=.8ex,inner sep=0em]

\usepackage{thmtools} 
\usepackage{thm-restate}
\usepackage{multicol}

\newcommand{\N}{\mathbb{N}}
\newcommand{\R}{\mathbb{R}}
\newcommand{\Q}{\mathbb{Q}}

\newcommand{\A}{\mathcal{A}}

\newcommand{\F}{\mathcal{F}}

\newcommand{\M}{\mathcal{M}}

\newcommand{\X}{\mathcal{X}}

\newcommand{\calL}{\mathcal{L}}

\newcommand{\calO}{\mathcal{O}}

\newcommand{\AP}{\textit{AP}}
\newcommand{\run}{\textit{Run}}
\newcommand{\NP}{\textbf{NP}}

\newcommand{\ex}[1]{\langle #1 \rangle}
\newcommand{\Ex}[1]{\langle\hspace*{-.7ex}\langle #1 \rangle\hspace*{-.7ex}\rangle}
\newcommand{\As}[1]{\mathtt{A}_{#1}}
\newcommand{\Bs}[1]{\mathtt{B}_{#1}}
\newcommand{\Cs}[1]{\mathtt{C}_{#1}}

\newcommand{\Init}{\mathit{Init}}
\newcommand{\bInit}{\textit{\textbf{Init}}}

\newcommand{\Abandon}{\mathit{Abandon}}
\newcommand{\Simulate}{\mathit{Sim}}

\newcommand{\Invariant}{\mathit{Invariant}}

\newcommand{\Free}{\mathit{Free}}
\newcommand{\bFree}{\textit{\textbf{Free}}}

\newcommand{\Fin}{\mathit{Fin}}
\newcommand{\bFin}{\textit{\textbf{Fin}}}
\newcommand{\Trans}{\mathit{Trans}}

\newcommand{\LTrans}{\mathit{LTrans}}
\newcommand{\CTrans}{\mathit{CTrans}}
\newcommand{\bCTrans}{\textit{\textbf{CTrans}}}

\newcommand{\bTrans}{\textit{\textbf{Trans}}}
\newcommand{\Transient}{\mathit{Transient}}
\newcommand{\LPass}{\mathit{LPass}}
\newcommand{\Step}{\mathit{Step}}
\newcommand{\NStep}{\mathit{STEP}}
\newcommand{\bStep}{\textit{\textbf{{Step}}}}
\newcommand{\Interval}{\mathit{Interval}}
\newcommand{\bInterval}{\textit{\textbf{Interval}}}
\newcommand{\Suc}{\mathit{Suc}}

\newcommand{\Fsuc}{\mathit{FSuc}}
\newcommand{\Zsuc}{\mathit{ZSuc}}
\newcommand{\Psuc}{\mathit{PSuc}}
\newcommand{\bZsuc}{\textit{\textbf{ZSuc}}}
\newcommand{\bPsuc}{\textit{\textbf{PSuc}}}
\newcommand{\Csuc}{\mathit{CSuc}}

\newcommand{\IZsuc}{\mathit{IZSuc}}
\newcommand{\IPsuc}{\mathit{IPSuc}}
\newcommand{\bIZsuc}{\textit{\textbf{IZSuc}}}
\newcommand{\bIPsuc}{\textit{\textbf{IPSuc}}}

\newcommand{\OZsuc}{\mathit{OZSuc}}
\newcommand{\OPsuc}{\mathit{OPSuc}}
\newcommand{\OZer}{\mathit{OZer}}
\newcommand{\OPos}{\mathit{OPos}}
\newcommand{\bFsuc}{\textit{\textbf{FSuc}}}
\newcommand{\Zero}{\mathit{Zero}}
\newcommand{\bZero}{\textit{\textbf{Zero}}}
\newcommand{\Eq}{\mathit{Eq}}
\newcommand{\bEq}{\textit{\textbf{Eq}}}
\newcommand{\Points}{\mathit{Points}}
\newcommand{\Line}{\mathit{Line}}
\newcommand{\Lines}{\mathit{LSegs}}
\newcommand{\Area}{\mathit{Area}}
\newcommand{\Out}{\mathit{Out}}
\newcommand{\Ins}{\mathit{Ins}}
\newcommand{\slope}{\mathit{slope}}

\newcommand{\Labels}{\mathit{Labels}}
\newcommand{\ta}{\mathit{ta}}
\newcommand{\tb}{\mathit{tb}}
\newcommand{\tc}{\mathit{tc}}
\newcommand{\te}{\mathit{tx}}
\newcommand{\Succ}{\mathit{succ}}

\renewcommand{\vec}[1]{\pmb{#1}}

\newcommand*{\vp}{\varphi}

\newcommand*{\opx}{\operatorname{\pmb{\mathtt{X}}}}
\newcommand*{\opu}{\operatorname{\pmb{\mathtt{U}}}}
\newcommand*{\opf}{\operatorname{\pmb{\mathtt{F}}}}
\newcommand*{\opg}{\operatorname{\pmb{\mathtt{G}}}}

\newcommand*{\m}{\mathbb{P}}

\begin{document}

\title[PCTL Satisfiability is Undecidable]{The General and Finite Satisfiability Problems\\ for PCTL are Undecidable}

\author{Miroslav Chodil}
\author{Anton{\'{\i}}n Ku{\v{c}}era}
\affiliation{%
  \institution{Faculty of Informatics, Masaryk University}
  \streetaddress{Botanick\'{a} 68a}
  \city{Brno}
  \country{Czechia}
  \postcode{60200}}








\renewcommand{\shortauthors}{Chodil and Ku\v{c}era}

\begin{abstract}
  The general/finite PCTL satisfiability problem asks whether a given PCTL formula has a general/finite model. We show that the finite PCTL satisfiability problem is \emph{undecidable}, and the general PCTL satisfiability problem is even \emph{highly undecidable} (beyond the arithmetical hierarchy). Consequently, there are no   
  sound deductive systems proving all generally/finitely valid PCTL formulae.
\end{abstract}



\received{20 February 2007}
\received[revised]{12 March 2009}
\received[accepted]{5 June 2009}

\maketitle

\section{Introduction}
\label{sec-intro}

 Probabilistic CTL (PCTL) \cite{HJ:logic-time-probability-FAC} is a temporal logic interpreted over states in discrete Markov chains. PCTL is obtained from the standard CTL (Computational Tree Logic, see, e.g., \cite{Emerson:temp-logic-handbook}) by replacing the existential/universal path quantifiers with the probabilistic operator $P(\Phi) \bowtie r$. Here,  $\Phi$ is a path formula, $\bowtie$ is a comparison such as ${\geq}$ or ${<}$, and $r$ is a rational numerical constant. A~formula $P(\Phi) \bowtie r$ holds in a state $s$ if the probability of all runs initiated in~$s$ satisfying~$\Phi$ is $\bowtie$-bounded by $r$.  

Unlike CTL and other non-probabilistic temporal logics, PCTL does not have the small model property guaranteeing the existence of a bounded-size model for every satisfiable formula. In fact, one can easily construct satisfiable PCTL formulae without \emph{any} finite model (see, e.g., \cite{BFKK:satisfiability}). Hence, the PCTL satisfiability problem is studied in two basic variants: (1) \emph{finite satisfiability}, where we ask about the existence of a finite model, and (2) \emph{general satisfiability}, where we ask about the existence of an unrestricted model. 

At first glance, the finite satisfiability problem appears simpler. Let $\varphi$ be a PCTL formula. The existence of a model for $\varphi$ with a \emph{given} number of states is decidable by encoding the question into first-order theory of the reals (see, e.g., \cite{ChK:PCTL-quatitative-fragments-JCSS}). Hence, the finite satisfiability problem is at least semi-decidable. To prove its decidability, it suffices to establish \emph{some} computable upper bound on the number of states of a model for a finite-satisfiable formula. One is tempted to conjecture the existence of such a bound, because there is no apparent way how a finite-satisfiable PCTL formula $\varphi$ can ``enforce'' the existence of $F(|\varphi|)$ distinct states in a model of $\varphi$, where $F$ grows faster than every computable function (such as the Ackermann function). Despite numerous research attempts resulting in positive decidability results for various PCTL fragments (see Related work), the decidability of general/finite PCTL satisfiability has remained open for almost 30~years.

\textbf{Our Contribution.} In this paper, we show that the general and the finite PCTL satisfiability problems are \emph{undecidable}.

The undecidability result for finite PCTL satisfiability holds even for a simple PCTL fragment consisting of formulae of the form \mbox{$\varphi_1 \wedge \opg_{=1} \varphi_2$}, where $\varphi_1,\varphi_2$ contain only the path connectives $\opx$ and $\opf^2$ ($\opx \psi$ says that $\psi$ holds in the next state, and $\opf^2 \psi$ says that $\psi$ holds in a state reachable in at most two steps).  An immediate consequence is that the \emph{finite validity} problem for PCTL is not even semi-decidable. Hence, there is no sound\,\&\,complete  deductive system proving all finitely valid PCTL formulae.

For general PCTL satisfiability, we show that the problem is even \emph{highly} undecidable ($\Sigma_1^1$-hard). This result holds even for a PCTL fragment consisting of formulae of the form \mbox{$\varphi_1 \wedge \opg_{=1} \varphi_2 \wedge \opg_{=1} \opf_{=1} \varphi_3$}, where $\varphi_1,\varphi_2$ contain only the path connectives $\opx$ and $\opf^2$, and $\varphi_3$ is a Boolean combination of atomic propositions. This implies that the (general) validity problem for PCTL is also highly undecidable.

\textbf{Paper Organization.}
The results are obtained by constructing a formula $\Psi$ simulating a computation of a given non-deterministic two-counter Minsky machine. The construction of $\Psi$ is based on combining several novel techniques. To make the construction comprehensible, we explain these techniques gradually and proceed in four main steps.
\begin{itemize}
    \item[(1)] In Section~\ref{sec-geom}, we introduce \emph{characteristic vectors} as a way of representing counter values, and two transformations $\tau,\sigma$ representing the \emph{decrement} and \emph{increment} operations on the counter (due to the chosen encoding, testing the counter for zero is trivial). 
    \item[(2)] In Section~\ref{sec-large}, we show that there exists a \emph{fixed} PCTL formula enforcing arbitrarily large finite models just by changing numerical constants $x$ and $y$ in its subformulae  $\opx_{=x} a$ and $\opx_{=y} b$. This result is perhaps interesting on its own because it reveals a specific power of probability constraints. 
    Intuitively, the constants $x$ and $y$ encode a characteristic vector representing a counter value $n$ (by choosing appropriate $x$ and $y$, the $n$ can be arbitrarily large), and the formula implements the function $\tau$ decrementing the counter. This enforces the existence of states whose characteristic vectors represent the counter values ranging from $n$ to $0$, and such states must be pairwise different.
    \item[(3)] In Section~\ref{sec-counter}, we extend the result of~(2) by constructing a formula $\psi$ simulating the computation of a non-deterministic one-counter Minsky machine. The crucial new ingredient is the construction implementing the \emph{increment} function~$\sigma$.
    \item[(4)] Finally, in Section~\ref{sec-Minsky-simulation}, we construct a formula $\Psi$ simulating a non-deterministic two-counter Minsky machine~$\M$. Technically, $\M$ is first ``translated'' into a synchronized product of two non-deterministic one-counter Minsky machines $\M_1$ and $\M_2$. Then, we use the formulae $\psi_1$ and $\psi_2$ constructed for $\M_1$ and $\M_2$ by the method of Section~\ref{sec-counter} and ``merge'' them into the formula $\Psi$.
\end{itemize}
In each step, we re-use the results of the previous steps, possibly after some necessary modifications. This leads to substantial simplifications at the cost of frequent references to previous sections. We compensate for this inconvenience by suggestive notation, writing the re-used formulae consistently in \textit{\textbf{boldface}}. The exact semantics of this notation is explained at the beginning of Section~\ref{sec-counter} and Section~\ref{sec-Minsky-simulation}.

The presented constructions bring some additional consequences formulated in Section~\ref{sec-concl}.


\textbf{Related Work.}
The probabilistic extension of CTL (and also CTL$^*$) has been initially studied in its \emph{qualitative} form, where the range of admissible probability constraints is restricted to $\{{=}0, {>}0, {=}1, {<}1\}$ \cite{LS:time-chance-IC,HS:Prob-temp-logic,KL:qPCTL-satisfiability}. Both general and finite satisfiability for qualitative PCTL are shown decidable in these works. A precise complexity classification of general and finite satisfiability for qualitative PCTL, together with a construction of (a finite description of) a model, are given in \cite{BFKK:satisfiability}. In the same paper, it is also shown that both general and finite satisfiability are undecidable when the class of admissible models is restricted to Markov chains with a $k$-bounded branching degree, where $k \geq 2$ is an arbitrary constant (this technique is not applicable to general Markov chains). A variant of the bounded satisfiability problem, where transition probabilities are restricted to $\{\frac{1}{2},1\}$, is proven \NP-complete
in \cite{BFS:bounded-PCTL}. 

The decidability of finite satisfiability for various quantitative PCTL fragments (with general probability constrains) is established in \cite{KR:PCTL-unbounded,CHK:PCTL-simple,ChK:PCTL-quatitative-fragments-JCSS}. More concretely, in \cite{CHK:PCTL-simple}, it is shown that every formula $\varphi$ of the \emph{bounded fragment} of PCTL, where the validity of~$\varphi$ in a state $s$ depends only on a bounded prefix of a run initiated in~$s$, has a bounded-size tree model. In \cite{KR:PCTL-unbounded}, several PCTL fragments based on $F$ and $G$ operators are studied. For each of these fragments, it is shown that every finite satisfiable formula has a bounded-size model where every non-bottom SCC is a singleton. In \cite{ChK:PCTL-quatitative-fragments-JCSS}, a more abstract decidability result based on isolating the progress achieved along a chain of visited SCCs is presented. 

The \emph{model-checking} problem for PCTL has been studied both for finite Markov chains (see, e.g., \cite{BK:PCTL-fairness,BA:MDP-PCTL,HK:quantitative-mu-calculus-LICS,BK:book}) and for infinite Markov chains generated by probabilistic pushdown automata and their subclasses \cite{EKM:prob-PDA-PCTL-LMCS,BKS:pPDA-temporal,EY:RMC-LTL-complexity-TCL}. PCTL formulae have also been used as \emph{objectives} in Markov decision processes (MDPs) and stochastic games, where the players controlling non-deterministic states strive to \mbox{satisfy/falsify} a given PCTL formula. Positive decidability results exist for finite MDPs and qualitative PCTL formulae \cite{BFK:MDP-PECTL-objectives}. For general PCTL and finite MDPs, the problem becomes undecidable~\cite{BBFK:Games-PCTL-objectives}. 

\section{Preliminaries}
\label{sec-prelim}

The sets of non-negative integers, rational numbers, and real numbers are denoted by $\N$, $\Q$, and $\R$, respectively. The intervals of real numbers are written in the standard way, e.g., $[0,1)$ is the set of all $r \in \R$ such that $0 \leq r < 1$.  

We use $\vec{v},\vec{u},\vec{\kappa},\ldots$ to denote the elements of $\R \times \R$. The first and the second components of $\vec{v}$ are denoted by $\vec{v}_1$ and $\vec{v}_2$, respectively.

The $n$-fold composition $f \circ \cdots \circ f$ of a function $f : A \to A$ (where $A$ is some set) is denoted by $f^n$. 

\subsection{The Logic PCTL}

The logic PCTL \cite{HJ:logic-time-probability-FAC} is obtained from the standard CTL (Computational Tree Logic \cite{Emerson:temp-logic-handbook}) by replacing the existential and universal path quantifiers with the probabilistic operator $P(\Phi) \bowtie r$, where $\Phi$ is a path formula, $\bowtie$ is a comparison, and $r \in [0,1]$ is a rational constant. 

\begin{definition}[PCTL]
\label{def-pctl}
   Let $\AP$ be a set of atomic propositions. The syntax of PCTL state and path formulae is defined by the following abstract syntax equations:
   \[
   \begin{array}{lcl}
      \varphi & ~~::=~~ & a \mid \neg \vp \mid \vp_1 \wedge \vp_2 \mid P(\Phi) \bowtie r\\
      \Phi & ::= &\opx \vp \mid \vp_1 \opu \varphi_2 \mid \vp_1 \opu^{k} \varphi_2
   \end{array} 
   \]    
   Here, $a \in \AP$, ${\bowtie} \in \{{\geq}, {>}, {\leq},{<},{=}\}$, $r \in [0,1]$ is a rational constant, and $k \in \N$.
\end{definition}
The formulae $\textit{true},\textit{false}$ and the other Boolean connectives are defined using $\neg$ and $\wedge$ in the standard way. We also use $\opf \vp$ and $\opf^k \vp$ to abbreviate the formulae $\textit{true} \opu \vp$ and $\textit{true} \opu^k \vp$, respectively. Furthermore, we often abbreviate a formula of the form $P(\Phi) \bowtie r$ by omitting $P$ and adjoining the probability constraint directly to the topmost path operator of $\Phi$. For example, we write $\opx_{{=}1} \varphi$ instead of $P(\opx\varphi) = 1$. We also write $\opg_{=1} \vp$ instead of $\opf_{=0} \neg\vp$.

PCTL formulae are interpreted over Markov chains where every state $s$ is assigned a subset $v(s) \subseteq \AP$ of propositions valid in~$s$.
  
\begin{definition}[Markov chain]
    A Markov chain is a triple $M = (S,P,v)$, where $S$ is a finite or countably infinite set of \emph{states}, \mbox{$P \colon S \times S \rightarrow [0,1]$} is a function such that $\sum_{t \in S} P(s,t)=1$ for every $s \in S$, and $v
    \colon S \rightarrow 2^{\AP}$ is a \emph{valuation}. We say that $M$ is finite if $S$ is a finite set.
\end{definition}

For $s,t \in S$, we say that $t$ is an \emph{immediate successor} of $s$ if \mbox{$P(s,t) > 0$}.
A \emph{path} in $M$ is a finite sequence $w = s_0, \ldots ,s_n$ of states where $n \geq 0$ and 
\mbox{$P(s_i,s_{i+1}) > 0$} for all $i <n$. We say that $t$ is \emph{reachable} from $s$ if there is a path where the first and the last state is~$s$ and~$t$, respectively.

A \emph{run} in $M$ is an infinite sequence $\pi = s_0, s_1, \ldots$ of states such that every finite prefix of $\pi$ is a path in $M$. We also use $\pi(i)$ to denote the state $s_i$ of~$\pi$. 



For every path $w = s_0, \ldots ,s_n$, let $\run(w)$ be the set of all runs starting with~$w$, and let $\m(\run(w)) = \prod_{i=0}^{n-1} P(s_i,s_{i+1})$. To every state $s$, we associate the probability space $(\run(s),\F_{s},\m_{s})$, where $\F_{s}$ is the \mbox{$\sigma$-field} generated by all $\run(w)$ where $w$ starts in $s$, and $\m_{s}$ is the unique probability measure obtained by extending $\m$ in the standard way (see, e.g., \cite{Billingsley:book}). 

The \emph{validity} of a PCTL state/path formula for a given state/run of $M$ is defined inductively as follows:
\[
\begin{array}{lcl}
  s \models a & \mbox{ ~~iff~~ } & a \in v(s),\\
  s \models \neg\varphi & \mbox{ ~~iff~~ } & s \not\models \varphi,\\
  s \models \vp_1 \wedge \vp_2 &  \mbox{ iff } & s \models \vp_1 \mbox{ and } s \models \vp_2,\\
  s \models P(\Phi) \bowtie r  &  \mbox{ iff } &  \m_{s}(\{ \pi \in \run(s) \mid \pi \models \Phi \}) \bowtie r,\\[1ex]
  \pi \models \opx \vp  &  \mbox{ iff } & \pi(1) \models \vp \mbox{ for some } i \in \N,\\
  \pi \models \vp_1 \opu \varphi_2 &  \mbox{ iff } & \mbox{there is } j\geq 0 \mbox{ such that }
\pi(j) \models \vp_2\\
  &&  \mbox{and }  \pi(i) \models \vp_1 \mbox{ for all } 0\leq i < j,\\
  \pi \models \vp_1 \opu^k \varphi_2 &  \mbox{ iff } & \mbox{there is } 0 \leq j\leq k \mbox{ such that } \pi(j) \models \vp_2\\
    &&  \mbox{and }  \pi(i) \models \vp_1 \mbox{ for all } 0\leq i < j.\\
\end{array} 
\]

We say that $M$ is a \emph{model} of $\varphi$ if $s \models \varphi$ for some state $s$ of $M$. 
The \emph{general/finite PCTL satisfiability problem} is the question of whether a given PCTL formula has a general/finite model.

\subsection{Parameterized PCTL Formulae}

A \emph{parameterized PCTL formula} is a PCTL formula where some probability constraints are replaced with \emph{parameters} ranging over rationals in $[0,1]$. For example, $\xi(x) \equiv \opf_{\geq 0.6} a \wedge \opg_{= x} \neg a$ is a parameterized PCTL formula with one parameter $x$. For a parameterized PCTL formula $\varphi(x_1,\ldots,x_k)$ and rational constants $p_1,\ldots,p_k$ in the interval $[0,1]$, we use $\varphi[p_1,\ldots,p_k]$ to denote the PCTL formula obtained from $\varphi(x_1,\ldots,x_k)$ by substituting every $x_i$ with $p_i$ (we say that $\varphi[p_1,\ldots,p_k]$  is an \emph{instance} of $\varphi(x_1,\ldots,x_k)$). For example, $\xi[0.1]$ is the formula  $\opf_{\geq 0.6} a \wedge \opg_{= 0.1} \neg a$.

\subsection{Minsky Machines}
\label{sec-Minsky}

A \emph{non-deterministic Minsky machine $\M$ with $k \geq 1$ counters} is a finite program
\[
    1: \Ins_1; \ \cdots  \ m: \Ins_m;
\]
where $m \geq 1$ and every $i: \Ins_i$ is a labeled instruction of one of the following types:
\begin{itemize}
  \item[I.] $i: \textit{inc } c_j; \textit{ goto } L;$ 
  \item[II.] $i: \textit{if } c_j{=}0 \textit{ then goto } L \textit{ else dec } c_j; \textit{ goto } L'$
\end{itemize}
Here, $j \in \{1,\ldots,k\}$ is a counter index and $L,L' \subseteq \{1,\ldots,m\}$
are sets of labels with one or two elements. We say that $\M$ is \emph{deterministic} if all $L,L'$ occurring in the instructions of $\M$ are singletons\footnote{Our definition of non-deterministic Minsky machines is equivalent to the standard one where the target sets of labels are singletons, and there is also a Type~III instruction of the form $i: \textit{goto } u \textit{ or } u'$. For purposes of this paper, the adopted definition is more convenient.}.

A \emph{configuration} of $\M$ is a tuple $(i,n_1,\ldots,n_k)$ of non-negative integers where $1 \leq i \leq m$ represents the current control position and $n_1,\ldots,n_k$ represent the current counter values. A configuration $(i',n_1',\ldots,n_k')$ is a \emph{successor} of a configuration  $(i,n_1,\ldots,n_k)$, written $(i,n_1,\ldots,n_k) \mapsto (i',n_1',\ldots,n_k')$, if the tuple  $(n_1',\ldots,n_k')$ is obtained from $(n_1,\ldots,n_k)$ by performing $\Ins_i$, and $i'$ is an element of the corresponding $L$ (or $L'$) in $\Ins_i$. Note that every configuration has either one or two successor(s). A \emph{computation} of $\M$ is an infinite sequence of configurations $\omega \equiv C_0,C_1,\ldots$ such that $C_0 = (1,0,\ldots,0)$ and $C_i \mapsto C_{i+1}$ for all $i \in \N$. We say that $\omega$ is \emph{periodic} if there are $i,j \in \N$ such that $i < j$ and the infinite sequences $C_i,C_{i+1},\ldots$ and $C_j,C_{j+1},\ldots$ are the same.

Now, we recall the standard undecidability results for Minsky machines. The symbols $\Sigma_1^0$ and $\Sigma_1^1$ denote the corresponding levels in the arithmetical and the analytical hierarchies, respectively.
\smallskip


(1) The \emph{boundedness problem} for a given deterministic two-counter Minsky machine $\M$ is undecidable and $\Sigma_1^0$-complete \cite{KSCh:Minsky-boundedness-PCS}. Here, $\M$ is bounded if the unique computation $\omega$ contains only finitely many pairwise different configurations (i.e., $\omega$ is periodic).
\smallskip

(2) The \emph{recurrent reachability problem} for a given non-deterministic \mbox{two-counter} Minsky machine $\M$ is highly undecidable a $\Sigma_1^1$-complete \cite{Harel:Infinite-trees-JACM}. Here, the question is whether there exists a \emph{recurrent} computation $\omega$ of~$\M$ such that the instruction $\Ins_1$ is executed infinitely often along~$\omega$.

\section{Representing a Counter}
\label{sec-geom}

In this section, we introduce several ``geometrical'' concepts underpinning our results. Furthermore, we show how to represent a non-negative counter value by a pair of quantities, and we design functions modeling the decrement/increment operation on the counter. Missing proofs are in the Appendix. 

Let us fix a rational constant $q$ such that $\frac{3}{4} < q < 1$ and $\sqrt{4q -3}$ is rational. For example, we can put $q = \frac{13}{16}$. Furthermore, we define
\[
I_q = \left(\frac{1-\sqrt{4q-3}}{2},  \frac{1+\sqrt{4q-3}}{2}\right)
\]
By our choice of $q$, we immediately obtain that $I_q \subseteq (0,1)$.
Finally, we fix $\vec{\kappa} \in (0,1)^2$ with rational components such that $\vec{\kappa}_1 \in I_q$ and $\vec{\kappa}_1 + \vec{\kappa}_2 \leq 1$. 

The rational constants $q$, $\vec{\kappa}_1$, and $\vec{\kappa}_2$ are used as probability constraints in the PCTL formulae constructed in the next sections. The defining properties of $q$ and $I_q$ are explained in Lemma~\ref{lem-tausigma}.

\begin{definition}[characteristic vector, $\As{t}$, $\Bs{t}$, and $\Cs{t}$ sets]
    Let $t$ be a state of a Markov chain with transition function~$P$. Let $\As{t}$, $\Bs{t}$, and $\Cs{t}$ be the sets of all immediate successors of $t$ satisfying the atomic propositions $a$, $b$, and $c$, respectively.

    The \emph{characteristic vector} of $t$ is the vector $\vec{v}[t] \in [0,1]^2$ where $\vec{v}[t]_1 = \sum_{u \in \As{t}} P(t,u)$ and $\vec{v}[t]_2 = \sum_{u \in \Bs{t}} P(t,u)$. 
\end{definition}

Observe that $\vec{v}[t]_1$ and $\vec{v}[t]_2$ is the probability of satisfying the path formula $\opx a$ and $\opx b$ in $t$, respectively.
Intuitively, we use characteristic vectors to encode non-negative integers, where $\vec{\kappa}$ represents zero, and the decrement/increment operations correspond to performing the functions $\tau$/$\sigma$ introduced in our next definition.   

\begin{definition}[$\tau$, $\sigma$, and $W$]
\label{def-W-t-s}
Let $W = I_q \times [0,\infty)$. Furthermore, let $\tau,\sigma : W \to \R^2$ be functions\footnote{$\tau$ and $\sigma$ are not arbitrary; they must satisfy several properties simultaneously to enable the presented constructions. There is no trivial intuition behind their design.} defined as follows: 
\begin{itemize}
    \item $\tau(\vec{v}) = \big( (q{-}1{+}\vec{v}_1)/\vec{v}_1,\, \vec{v}_2/\vec{v}_1 \big)$ 
    \item $\sigma(\vec{v}) = \big( (1{-}q)/(1{-}\vec{v}_1), (\vec{v}_2(1{-}q))/(1{-}\vec{v}_1)  \big)$.
\end{itemize}  
\end{definition}

The \emph{slope} of a line or a line segment in $\R^2$ is defined in the standard way. For all $\vec{v},\vec{u} \in \R^2$ where $\vec{v}_1 \neq \vec{u}_1$, we use $\slope(\vec{v},\vec{u})$ to denote the slope of the line containing $\vec{v},\vec{u}$.

In the next lemma, we use the defining properties of $q$ and $I_q$ (this explains their purpose).

\begin{restatable}{lemma}{tausigma}
\label{lem-tausigma}
For every $\vec{v} \in W$, we have the following: 
\begin{itemize}
    \item[(a)] $\tau(\vec{v}),\sigma(\vec{v}) \in W$;
    \item[(b)] $\tau(\vec{v})_1 > \vec{v}_1$ and $\tau(\vec{v})_2 \geq \vec{v}_2$; if $\vec{v}_2 > 0$, then $\tau(\vec{v})_2 > \vec{v}_2$; 
    \item[(c)] let $\vec{u} = (\vec{v}_1,0)$; then $\slope(\vec{u},\tau(\vec{v})) = \slope(\tau(\vec{v}),\tau^2(\vec{v}))$;
    \item[(d)] let $\vec{u} = (\vec{v}_1,y)$ where $0 \leq y <  \vec{v}_2$. Then $\slope(\vec{u},\tau(\vec{u})) < \slope(\vec{v},\tau(\vec{v}))$;
    \item[(e)] $\sigma(\tau(\vec{v})) = \tau(\sigma(\vec{v})) = \vec{v}$.
\end{itemize}
\end{restatable}


For every $\vec{u} \in W$, let $L(\vec{u})$ be the line segment between the points $\vec{u}$ and $\tau(\vec{u})$, including $\vec{u}$ and excluding $\tau(\vec{u})$, i.e., 
\[
    L(\vec{u}) = \{\vec{w} \in \R^2 \mid \vec{w} = \lambda \vec{u} + (1{-}\lambda)\tau(\vec{u}) \mbox{ for some } \lambda \in (0,1] \}.
\]

\begin{figure}[t]\centering
    \begin{tikzpicture}[x=1.5cm, y=1.2cm,font=\small]   
    \path[thick,-stealth, at start] 
        (-0.2,0) edge  (5.5,0)
        (0,-0.2) edge node[below] {$0$} (0,6);
        \path[thick,-, at start]
            (5,-.1)  edge node[below] {$1$} (5,0.1);
        \path[thick, dotted, at start] 
            (.8,-0.2) edge node[below] {$\frac{1-\sqrt{4q-3}}{2}$} (0.8,6)
            (4.2,-0.2) edge node[below] {$\frac{1+\sqrt{4q-3}}{2}$} (4.2,6);  
    \tikzset{ver/.style={draw, circle, inner sep=0pt, minimum size=.15cm, fill=black}}
    \coordinate (a) at (1.4,.1);
    \coordinate (b) at (1.9,.25);
    \coordinate (b0) at (1.9,0);
    \coordinate (c)  at (2.5,.7);
    \coordinate (c0) at (2.5,0);
    \coordinate (d) at (3.1,1.4);
    \coordinate (d0) at (3.1,0);
    \coordinate (e) at (3.65,2.8);
    \coordinate (e0) at (3.65,0);
    \coordinate (f) at (3.95,4.3);
    \coordinate (s) at (.8,0);
    \coordinate (t) at (4.1,6);
    \coordinate (tt) at (0.8,6);
    \draw [fill=gray!20,draw opacity=0] plot coordinates {(s) (a) (b) (c) (d) (e) (f) (t) (tt)};
    \draw[thick, dotted] (b) -- (b0) -- (c)-- (c0) -- (d)-- (d0) -- (e)-- (e0); 
    \node[ver] (A) at (a) {};
    \node[ver, label={above, xshift=-3mm, yshift=-1.5mm: $\sigma(\vec{v})$}] (B) at (b) {};
    \node[ver, label={above, xshift=-2mm, yshift=-1mm: $\vec{v}$}] (C) at (c) {};
    \node[ver, label={above, xshift=-3mm, yshift=-1mm: $\tau(\vec{v})$}] (D) at (d) {};
    \node[ver, label={above, xshift=-4mm, yshift=-1mm: $\tau^2(\vec{v})$}] (E) at (e) {};
    \draw[thick] (a) -- (b) -- (c) -- (d) -- (e);
    \draw[thick,dashed] (s) -- (a);
    \draw[thick,dashed] (e) -- (f) -- (t);
    \node (T) at (2,5) {$\Area(\vec{v})$};
\end{tikzpicture}
\caption{$\Points(\vec{v})$, $\Lines(\vec{v})$, and $\Area(\vec{v})$.}
\label{fig-area}
\end{figure}

We use $\Line(\vec{u})$ to denote the line obtained by prolonging the line segment $L(\vec{u})$, and $H(\vec{u})$ to denote the closed half-space above $\Line(\vec{u})$, i.e.,
\[
    H(\vec{u}) = \{\vec{u}+\vec{\alpha} \mid \vec{\alpha} \in \R^2, (\tau(\vec{u}_2){-}\vec{u}_2,\vec{u}_1{-}\tau(\vec{u})_1)\cdot \vec{\alpha} \leq 0\}.
\]
\begin{definition}
   For every $\vec{v} \in W$ where $\vec{v}_2 > 0$, let 
   \begin{eqnarray*}
    \Points(\vec{v}) & = & \{\tau^k(\vec{v}), \sigma^k(\vec{v}) \mid k \in \N \},\\[1ex]
    \Lines(\vec{v}) & = & \bigcup_{\vec{u}\in \Points(\vec{v})} L(\vec{u}),\\[1ex]
    \Area(\vec{v}) & = & W \cap \bigcap_{\vec{u} \in \Points(\vec{v})} H(\vec{u}).
   \end{eqnarray*}
\end{definition}

\noindent
The structure of $\Points(\vec{v})$, $\Lines(\vec{v})$, and $\Area(\vec{v})$ is shown in Fig.~\ref{fig-area}, where the dotted lines illustrate the property of Lemma~\ref{lem-tausigma}~(c).


Now, we present a sequence of technical observations culminating with (crucial) Theorem~\ref{thm-area}.

A convex combination of vectors $\vec{u^1},\vec{u^2},\ldots$ is \emph{positive} if all coefficients used in the combination are positive. The next lemma is a trivial corollary to Lemma~\ref{lem-tausigma} (see Fig.~\ref{fig-area}).

\begin{restatable}{lemma}{gproperties}
\label{lem-gproperties}
For every $\vec{v} \in W$ where $\vec{v}_2 > 0$ and every $\vec{u} \in \Points(\vec{v})$, we have the following:
\begin{itemize}
    \item[(a)] If $\vec{u}$ is a positive convex combination of $\vec{u^1},\vec{u^2},\ldots$ where $\vec{u^i} \in \Area(\vec{v})$ for all $i \in \N$, then $\vec{u^i} = \vec{u}$ for all $i \in \N$.
    \item[(b)] If $\vec{w} \in L(\vec{u})$ is a positive convex combination of 
    $\vec{u^1},\vec{u^2},\ldots$ where $\vec{u^i} \in \Area(\vec{v})$ for all $i \in \N$, then $\vec{u^i} \in L(\vec{u}) \cup \{\tau(\vec{u})\}$ for all $i \in \N$.
\end{itemize} 
\end{restatable}

    
\begin{restatable}{lemma}{tauconvex}
\label{lem-tau-convex}
Let $\vec{w} \in W$ and $\vec{u} \in L(\vec{w})$. Then $\tau(\vec{u}) \in L(\tau(\vec{w}))$.
\end{restatable}

\begin{restatable}{lemma}{outlineseg}
    \label{lem-outlineseg}
    For all $\vec{v} \in W \smallsetminus \Area(\vec{\kappa})$ where $\vec{v}_1 \leq \vec{\kappa}_1$, there exists $\vec{u} \in W \smallsetminus \Area(\vec{\kappa})$ such that $\vec{u}_1 = \sigma^k(\vec{\kappa})_1$ for some $k \geq 0$ and $\vec{v} \in L(\vec{u})$.
\end{restatable}


Now we prove the main result of this section. 


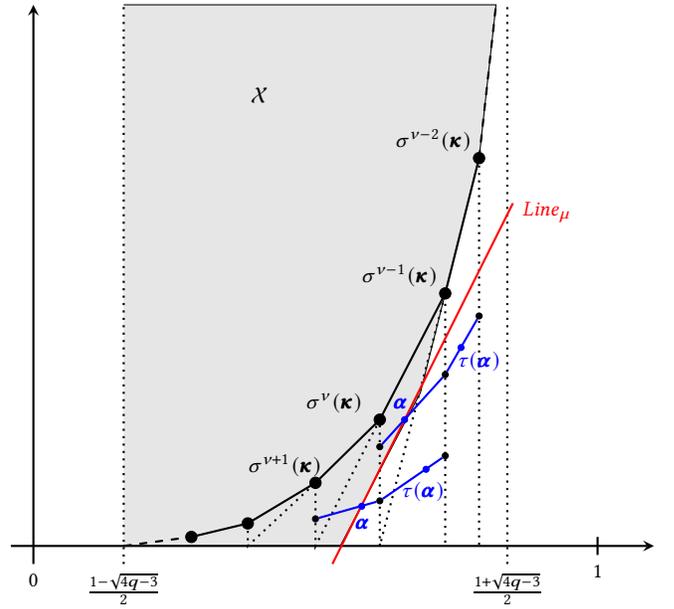
\begin{figure}[t]\centering
    \begin{tikzpicture}[x=1.5cm, y=1.2cm,font=\small]
    \tikzset{ver/.style={draw, circle, inner sep=0pt, minimum size=.15cm, fill=black}}    
    \tikzset{vr/.style={draw, circle, inner sep=0pt, minimum size=.08cm, fill=black}}    
    \coordinate (a) at (1.4,.1);
    \coordinate (b) at (1.9,.25);
    \coordinate (b0) at (1.9,0);
    \coordinate (c)  at (2.5,.7);
    \coordinate (c0) at (2.5,0);
    \coordinate (d) at (3.07,1.4);
    \coordinate (d0) at (3.07,0);
    \coordinate (e) at (3.65,2.8);
    \coordinate (e0) at (3.65,0);
    \coordinate (f) at (3.95,4.3);
    \coordinate (f0) at (3.95,0);
    \coordinate (s) at (0.8,0);
    \coordinate (t) at (4.1,6);
    \coordinate (tt) at (0.8,6);
    \draw [fill=gray!20,draw opacity=0] plot coordinates {(s) ([shift={(0.22,0)}]c0) ([shift={(0.36,1.7)}]d0) (e) (f) (t) (tt)};
    \path[thick,-stealth, at start] 
    (-0.2,0) edge  (5.5,0)
    (0,-0.2) edge node[below] {$0$} (0,6);
    \path[thick,-, at start]
        (5,-.1)  edge node[below] {$1$} (5,0.1);
    \path[thick, dotted, at start] 
        (.8,-0.2) edge node[below] {$\frac{1-\sqrt{4q-3}}{2}$} (0.8,6)
        (4.2,-0.2) edge node[below] {$\frac{1+\sqrt{4q-3}}{2}$} (4.2,6);   
    \draw[thick, dotted] (b) -- (b0) -- (c) -- (c0) -- (d) -- (d0) -- (e) -- (e0);
    \draw[thick, dotted] (f) -- (f0); 
    \node[ver] (A) at (a) {};
    \node[ver] (B) at (b) {};
    \node[ver, label={above, xshift=-4mm, yshift=-1mm: $\sigma^{\nu+1}(\vec{\kappa})$}] (C) at (c) {};
    \node[ver, label={above, xshift=-6mm, yshift=-1mm: $\sigma^{\nu}(\vec{\kappa})$}] (D) at (d) {};
    \node[ver, label={above, xshift=-6mm, yshift=-1mm: $\sigma^{\nu-1}(\vec{\kappa})$}] (E) at (e) {};
    \node[ver, label={above, xshift=-6mm, yshift=-1mm: $\sigma^{\nu-2}(\vec{\kappa})$}] (F) at (f) {};
    \draw[thick] (a) -- (b) -- (c) -- (d) -- (e) -- (f);
    \draw[thick,dashed] (s) -- (a);
    \draw[thick,dashed] (f) -- (t);
    \draw[thick,color=red] ([shift={(0.15,-0.2)}]c0) -- ([shift={(0.6,1.0)}] e);
    \node[thick,color=red] at ([shift={(.9,.9)}] e) {$\Line_\mu$};
    \node[vr] (r1) at ([shift={(0,-.4)}]c) {};
    \node[vr] (r2) at ([shift={(0,-.9)}]d) {};
    \node[vr] (r3) at ([shift={(0,-1.8)}]e) {};
    \draw[thick,color=blue] (r1) -- (r2) -- (r3);   
    \node[vr,color=blue,label={below, color=blue, xshift=.1mm, yshift=-.1mm: $\vec{\alpha}$}] (rho1) at ([shift={(.41,-.26)}]c) {};
    \node[vr,color=blue,label={below, color=blue, xshift=-.3mm, yshift=-.1mm: $\tau(\vec{\alpha})$}] (trho1) at ([shift={(.41,-.55)}]d) {};
    \node[vr] (rr1) at ([shift={(0,-.3)}]d) {};
    \node[vr] (rr2) at ([shift={(0,-.9)}]e) {};
    \node[vr] (rr3) at ([shift={(0,-1.75)}]f) {};
    \draw[thick,color=blue] (rr1) -- (rr2) -- (rr3);   
    \node[vr,color=blue,label={above, color=blue, xshift=-.5mm, yshift=-.1mm: $\vec{\alpha}$}] (rho2) at 
         ([shift={(.22,0)}]d) {};
    \node[vr,color=blue,label={below, color=blue, xshift=2.5mm, yshift=1mm: $\tau(\vec{\alpha})$}] (trho2) at 
         ([shift={(.14,-.6)}]e) {}; 
    \node (T) at (2,5) {$\X$};
\end{tikzpicture}
\caption{The construction proving $\Out = \emptyset$.}
\label{fig-thm-area}
\end{figure}

\begin{theorem}
\label{thm-area}
Let $T$ be a subset of states of some Markov chain with transition function $P$ such that for every $t \in T$, we have that $\vec{v}[t]_1 \leq \vec{\kappa}_1$ and if $\vec{v}[t] \neq \vec{\kappa}$, then  $\As{t} \subseteq T$ and the following equations are satisfied:
\begin{eqnarray}
q & = & 1 - \vec{v}[t]_1 + \sum_{u \in \As{t}} P(t,u)\cdot \vec{v}[u]_1
   \label{eq-convex-one}\\
q & = &  1 - \vec{v}[t]_1 - \vec{v}[t]_2 +  \sum_{u \in \As{t}} P(t,u)\cdot (\vec{v}[u]_1 + 
         \vec{v}[u]_2)
   \label{eq-convex-two}  
\end{eqnarray}
Then $\vec{v}[t] \in \Area(\vec{\kappa})$ for every $t \in T$. Furthermore, for every $t \in T$ such that $\vec{v}[t] \neq \vec{\kappa}$ and every $u \in \As{t}$, we have that $\vec{v}[u] = \tau(\vec{v}[t])$.
\end{theorem}
\begin{proof}
   We start be proving the following claim: for every $t \in T$ such that  \mbox{$\vec{v}[t] \neq \vec{\kappa}$}, the vector $\tau(\vec{v}[t])$ is a positive convex combination of the vectors in $\{\vec{v}[u] \mid u \in \As{t}\}$.
 
   By rewriting~\eqref{eq-convex-one}, we obtain    
    \begin{equation}
        \frac{q-1+\vec{v}[t]_1}{\vec{v}[t]_1} \ = \ 
        \sum_{u \in U_t} \frac{P(t,u)}{\vec{v}[t]_1} \vec{v}[u]_1   
        \label{eq-convex-first}
    \end{equation}
    Note that the left-hand side of \eqref{eq-convex-first} is equal to $\tau(\vec{v}[t])_1$. 
    Furthermore, by simplifying the right-hand side of \eqref{eq-convex-two} using ~\eqref{eq-convex-one}, we obtain
    \begin{equation}
        q \ = \ q  - \vec{v}[t]_2 + \sum_{u \in U_t} P(t,u) \cdot \vec{v}[u]_2
    \label{eq-convex-22}
    \end{equation}
    Thus,
    \begin{equation}
        \frac{\vec{v}[t]_2}{\vec{v}[t]_1} \ = \ \sum_{u \in U_t} \frac{P(t,u)}{\vec{v}[t]_1}\cdot \vec{v}[u]_2
    \label{eq-convex-second}
    \end{equation}
    Note that the left-hand side of \eqref{eq-convex-second} is equal to $\tau(\vec{v}[t])_2$. By combining \eqref{eq-convex-first} and \eqref{eq-convex-second}, we have that 
    \begin{equation}
        \tau(\vec{v}[t]) \ = \ \sum_{u \in U_t} \frac{P(t,u)}{\vec{v}[t]_1}\cdot \vec{v}[u]
    \end{equation}
    which proves the claim.

    Now we prove the first part of the theorem, i.e., $\vec{v}[t] \in \Area(\vec{\kappa})$ for every $t \in T$. Suppose the converse, i.e., the set $\Out$ consisting of all $t \in T$ such that $\vec{v}[t] \not\in \Area(\vec{\kappa})$ is non-empty. We show that this assumption leads to a contradiction. The arguments are illustrated in Fig.~\ref{fig-thm-area}.
    
    Let $\vec{v}[\Out] = \{\vec{v}[t] \mid t \in \Out\}$, and let
    $\nu$ be the \emph{minimal} \mbox{$k \in \N$} such that $\vec{v}[\Out] \smallsetminus H(\sigma^k(\vec{\kappa})) \neq \emptyset$. Furthermore, for every \mbox{$0 \leq y \leq \sigma^{\nu}(\vec{\kappa})_2$}, let $\Line_y$ be the line with the same slope as $\Line(\sigma^\nu(\vec{\kappa}))$ containing the point $(\sigma^\nu(\vec{\kappa})_1, y)$. We use $H(\Line_y)$ to denote the closed half-space above $\Line_y$. 
    Let $\mu$ be the supremum of all \mbox{$y \leq \sigma^{\nu}(\vec{\kappa})_2$} such that
    $\vec{v}[\Out] \subseteq H(\Line_y)$. Clearly, $0 \leq \mu < \sigma^\nu(\vec{\kappa})_2$. 
    Let 
    \[
       \X \quad = \quad W \ \cap \ H(\Line_\mu) \ \cap \ \bigcap_{k=0}^{\nu-1} H(\sigma^k(\vec{\kappa}))
    \]
    Note $\X$ is a convex set and $\Area(\vec{\kappa}) \cup \vec{v}[\Out] \subseteq \X$.  We show that there exists $t \in \Out$ such that $\tau(\vec{v}[t]) \not\in \X$. By the claim proven above, $\tau(\vec{v}[t])$ is a positive convex combination of the vectors in $\{\vec{v}[u] \mid u \in \As{t}\}$. However, this is impossible because all these vectors are in $\X$, and $\X$ is a convex set; we have a contradiction.

    The existence of $t \in \Out$ such that $\tau(\vec{v}[t]) \not\in \X$ is proven as follows. By the definition of $\mu$, there is an infinite sequence $t_1,t_2,\ldots$ such that 
    \begin{itemize}
        \item $t_i \in \Out$ and $\vec{v}[t_i]_1 \leq \vec{\kappa}_1$ for all $i \in \N$,
        \item the distance of $\vec{v}[t_i]$ from $\Line_\mu$ approaches zero as $i \to \infty$.
    \end{itemize}
    This sequence must contain an infinite subsequence converging to some point $\vec{\alpha} \in \Line_\mu \cap \X$.

    Since $\tau$ is continuous, it suffices to show that $\tau(\vec{\alpha}) \not\in \X$ for every $\vec{\alpha} \in \Line_\mu \cap \X$. So, let us fix some $\vec{\alpha} \in \Line_\mu \cap \X$. By Lemma~\ref{lem-outlineseg}, there exists $\vec{u} \in W \smallsetminus \Area(\kappa)$ such that $\vec{u}_ 1 = \sigma^k(\vec{\kappa})_1$ for some $k \geq 0$ and $\vec{\alpha} \in L(\vec{u})$. By the definition of $\Line_{\mu}$, we have that $k$ is either $\nu{+}1$ or $\nu$, depending on whether $\vec{\alpha}_1 < \sigma^{\nu}(\vec{\kappa})_1$ or not, respectively. In the first case, the slope of both $L(\vec{u})$ and $L(\tau(\vec{u}))$ is strictly smaller that the slope of $\Line_\mu$ (see Lemma~\ref{lem-tausigma}). This implies $\tau(\vec{u}) \not\in \X$ and $L(\tau(\vec{u})) \cap \X = \emptyset$, hence  $\tau(\vec{\alpha}) \not\in \X$ because $\tau(\vec{\alpha}) \in L(\tau(\vec{u}))$ by Lemma~\ref{lem-tau-convex}. In the second case, the slope of $L(\vec{u})$ 
    strictly smaller than the slope of $\Line_\mu$, which implies $\tau(\vec{u}) \not\in \X$. Furthermore, the slope of $L(\tau(\vec{u}))$ is strictly smaller than the slope of
    $L(\sigma^{\nu-1}(\vec{\kappa}))$, and hence $L(\tau(\vec{u})) \cap \X = \emptyset$. Since $\tau(\vec{\alpha}) \in L(\tau(\vec{u}))$ by Lemma~\ref{lem-tau-convex}, we obtain 
    $\tau(\vec{\alpha}) \not\in \X$.
    
    The second part of the theorem follows easily. Let $t \in T$ such that $\vec{v}[t] \neq \vec{\kappa}$. By the above claim, $\tau(\vec{v}[t])$ is a positive convex combination of the vectors in $\{\vec{v}[u] \mid u \in \As{t}\}$ and all of these vectors belong to $\Area(\vec{\kappa})$ by the first part of the theorem. Hence, $\vec{v}[u] = \tau(\vec{v}[t])$ for all $u \in \As{t}$ by applying \mbox{Lemma~\ref{lem-gproperties}~(a)} 
\end{proof}

\tikzstyle{state}=[draw,thick,minimum width=12mm,minimum height=12mm,rounded corners,text centered]
\begin{figure*}[t]\centering
\begin{tikzpicture}[x=1.7cm, y=1.7cm,scale=0.83,font=\small]
\node (t) at (0,0) [state, draw=none] {\LARGE$\bullet$};
\node [below of = t,node distance=3ex] (pt) {$\ex{a,r_i}$};
\node [above of = t, node distance=2ex] (lt) {$t$};
\begin{scope}[shift={(8.5,-1)}]
    \node (t1) at (0,0) [state] {};
    \node [below of = t1] (pt1) {$\ex{h,b,S^2(r_i)}$};
    \node (t2) at (2,0) [state] {};
    \node [below of = t2] (pt2) {$\ex{h,c,S^2(r_i)}$};
\end{scope}
\begin{scope}[shift={(5.5,0.5)}]
    \node (u) at (-2,-1.5) [state] {\LARGE$\bullet\ \cdots \ \bullet$};
    \node [above of = u, node distance=2ex] (lu) {$u_1\hspace{5.5ex}u_k$};
    \node [below of = u] (pu) {$\ex{a,S(r_i)}$};
    \node (z1) at (-5.2,-3) [state] {};
    \node [below of = z1] (pz1) {$\ex{a,S^2(r_i)}$};
    \node (z2) at (-4,-3) [state] {};
    \node [below of = z2] (pz2) {$\ex{h,b,S^3(r_i)}$};
    \node (z3) at (-2.8,-3) [state] {};
    \node [below of = z3] (pz3) {$\ex{h,c,S^3(r_i)}$};
    \node (zk1) at (-.7,-3) [state] {};
    \node [below of = zk1] (pzk1) {$\ex{a,S^2(r_i))}$};
    \node (zk2) at (0.5,-3) [state] {};
    \node [below of = zk2] (pzk2) {$\ex{h,b,S^3(r_i)}$};
    \node (zk3) at (1.7,-3) [state] {};
    \node [below of = zk3] (pzk3) {$\ex{h,c,S^3(r_i)}$};
    \draw [thick, dotted] ($(z3) +(.7,0)$)  -- ($(zk1) +(-.7,0)$);
\end{scope}
\draw[-stealth,rounded corners,thick] (t.center) -- ++(8.5,0) --  node[left] {$\vec{v}[t]_2$} (t1);
\draw[-stealth,rounded corners,thick] (t.center) -- ++(10.5,0) -- node[left] {$1{-}\vec{v}[t]_1{-}\vec{v}[t]_2$} (t2);
\draw[-stealth,rounded corners,thick] (t.center) -| node[left, near end] {$\vec{v}[t]_1$} (u);
\draw[-stealth,rounded corners,thick] ([xshift=4ex]u.center) -| node[left, near end] {$\vec{v}[u_k]_1$} (zk1);
\draw[-stealth,rounded corners,thick] ([xshift=4ex]u.center) -| node[left, near end] {$\vec{v}[u_k]_2$} (zk2);
\draw[-stealth,rounded corners,thick] ([xshift=4ex]u.center) -| node[left, near end] {} (zk3);
\draw[-stealth,rounded corners,thick] ([xshift=-4ex]u.center) -| node[left, near end] {$\vec{v}[u_1]_1$} (z1);
\draw[-stealth,rounded corners,thick] ([xshift=-4ex]u.center) -| node[left, near end] {$\vec{v}[u_1]_2$} (z2);
\draw[-stealth,rounded corners,thick] ([xshift=-4ex]u.center) -| node[left, near end] {} (z3);
\end{tikzpicture}
\caption{The structure of transient states in a model of $\psi[c,d]$}
\label{fig-transient}
\end{figure*}
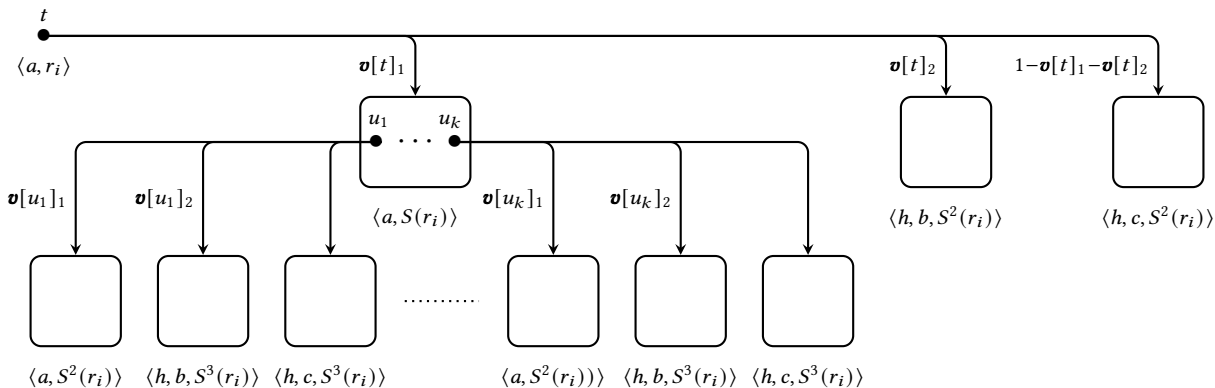

\section{A Parameterized PCTL Formula with Arbitrarily Large Models}
\label{sec-large}

In this section, we prove the following result:

\begin{theorem}
\label{thm-paramaterized}
There exists a parameterized PCTL formula $\psi(x,y)$ such that for every $n \in \N$, there exists an instance $\psi[c,d]$ satisfying the following:
\begin{itemize}
    \item every model of $\psi[c,d]$ has at least $n$ states;
    \item $\psi[c,d]$ has a finite-state model with $\calO(n)$ states.
\end{itemize}
\end{theorem}

\subsection{Constructing $\psi(x,y)$}
\label{sec-large-formula}

Let us fix $q$ and $\vec{\kappa}$ in the same way as in Section~\ref{sec-geom}. 
The set of atomic propositions\footnote{For purposes of this section, it suffices to use only three $r_i$ propositions instead of five (this is apparent when inspecting Fig.~\ref{fig-transient}). However, using five $r_i$'s allows us to reuse some of the constructed formulae in the next sections with only trivial modifications.} occurring in $\psi(x,y)$ is $A = \{a,b,c,h,r_0,r_1,r_2,r_3,r_4\}$.  For every $B \subseteq A$, we use $\ex{B}$ to denote the formula saying that \emph{exactly} the propositions of $B$ are satisfied, i.e.,  $\ex{B} \equiv \bigwedge_{p \in B} p \wedge \bigwedge_{p \in A \smallsetminus B} \neg p$. Slightly abusing our notation, we write, e.g., $\ex{a,r_i}$ instead of $\ex{\{a,r_i\}}$.
Furthermore, for every $r_i \in A$, we use $S(r_i)$ to denote the ``successor'' proposition $r_j \in A$ such that $j = i{+}1~\mathrm{mod}~5$. For example, $S(S(r_3)) = r_0$. Recall that $S^k$ denotes the $k$-fold composition $S\circ \cdots \circ S$.

We put
\[
    \psi(x,y) \ \equiv \ \Init(x,y) \ \wedge \ \opg_{=1} \Invariant 
\]
where
\[
    \Init(x,y) \ \equiv \ \ex{a,r_0} \ \wedge \ \opx_{=x} a \ \wedge \  \opx_{=y} b
\]
is a parameterized \emph{initial condition} that has to be valid in a state $s$ satisfying an instance of $\psi(x,y)$, and 
\[
    \Invariant \ \equiv \ \Fin \vee \Trans \vee \Free
\]
is a formula (with no parameters) that must be valid in every state reachable from~$s$. The formula 
\[
    \Free \ \equiv \ h \wedge \bigvee_{B \subseteq A} (\ex{B} \wedge \opx_{=1} \ex{B})
\]
ensures that every reachable state $t$ satisfying the predicate $h$ has only immediate successors satisfying the same subset of $A$ as~$t$. This enforces that $\Free$ is valid in all states reachable from~$t$ (intuitively, these states are ``free'' in the sense that they do not require further attention).

The formula $\Fin$ is defined as follows:
\[
    \Fin \ \equiv \ \bigvee_{i \in \{0,\ldots,4\}}  \ex{a,r_i} \wedge \Fsuc_i \wedge \Zero
\]
where
\begin{eqnarray*}
    \Fsuc_i & \equiv &  \opx_{=1} (\ex{h,a,S(r_i)} \vee \ex{h,b,S^2(r_i)} \vee \ex{h,c,S^2(r_i)}),\\   
    \Zero & \equiv &  \opx_{=\vec{\kappa}_1} a \ \wedge \  \opx_{=\vec{\kappa}_2} b  \,.
\end{eqnarray*}  
Finally, we put 
\[
    \Trans \ \equiv \ \bigvee_{i \in \{0,\ldots,4\}} \ex{a,r_i} \wedge \Suc_i \wedge \Interval \wedge \Eq_i
\]
where
\begin{eqnarray*}
    \Suc_i & \equiv & \opx_{=1} (\ex{a,S(r_i)} \vee \ex{h,b,S^2(r_i)} \vee \ex{h,c,S^2(r_i)})\\
    \Interval & \equiv & \opx_{>\left(1-\sqrt{4q-3}\right)/2} a \ \wedge \ 
                      \opx_{\leq \vec{\kappa}_2} a \ \wedge \ \opx_{>0} b\\
    \Eq_i  & \equiv & \opf^{2}_{=q} S^2(r_i) \ \wedge \ 
                      \opf^{2}_{=q} ((S^2(r_i) \wedge \neg b) \vee (S^3(r_i) \wedge b))
\end{eqnarray*}  
This completes the definition of $\psi(x,y)$.

\subsection{A Proof of Theorem~\ref{thm-paramaterized}}
\label{sec-proof-main-param}

For a given $n \in \N$, let $c = \sigma^n(\vec{\kappa})_1$, $d = \sigma^n(\vec{\kappa})_2$, and consider the instance $\psi[c,d]$ of $\psi(x,y)$ (note that $c,d$ are rational). We show that $\psi[c,d]$ satisfies the two claims of Theorem~\ref{sec-proof-main-param}.

To prove the first claim, let us fix a Markov chain $M = (S,P,v)$ 
such that $s \models \psi[c,d]$ for some $s \in S$. If $S$ is infinite, we are done immediately. We show that if $S$ is finite, then $S$ contains at least $n$~states. Without restrictions, we assume that all states of $S$ are reachable from~$s$. 

Note that for every $t \in S$, we have that \emph{exactly one}\footnote{In particular, note that assuming $t \models \Fin \wedge \Trans$ implies $t \models \opx_{>0} \ex{a,S(r_i)}$ and $t \models \opx_{>0} \ex{h,a,S(r_i)}$, hence $t \not\models \Fsuc_i$ and $t \not\models \Suc_i$ (contradiction).}  of the following possibilities holds: $t \models \Free$, $t \models \Fin$, or $t \models \Trans$.
Thus, $S$ is partitioned into three disjoint subsets of \emph{free}, \emph{final}, and \emph{transient} states, respectively. 

\subsubsection*{Properties of free states.}
Observe that a state $t$ is free iff $t \models h$. The formula $\Free$ ensures that all immediate successors of a free state~$t$ (and consequently also all states reachable from $t$) satisfy the same set of atomic propositions as~$t$.

\subsubsection*{Properties of final states.}
The formula $\Fin$ ensures that all immediate successors of every final state are free and the characteristic vector of a final state is equal to $\vec{\kappa}$.

\subsubsection*{Properties of transient states.}
Let $t$ be a transient state, i.e.,
\mbox{$t \models \ex{a,r_i} \wedge \Suc_i \wedge \Interval \wedge \Eq_i$}
for exactly one $i \in \{0,\ldots,4\}$.
The formula $\Suc_i$ says that $t$ has precisely three types of immediate successors:
\begin{itemize}
    \item[I.]   transient or final states satisfying $\ex{a,S(r_i)}$,
    \item[II.]  free states satisfying $\ex{h,b,S^2(r_i)}$,
    \item[III.] free states satisfying $\ex{h,c,S^2(r_i)}$.
\end{itemize}
Note that the first/second component of $\vec{v}[t]$ is precisely the total probability of entering an immediate successor of Type~I/II from~$t$, respectively (see Fig.~\ref{fig-transient}, where $u_1,\ldots,u_k$ are Type~I successors of~$t$, and $\vec{v}[t]_1 = \sum_{j=1}^k P(t,u_i)$).

Let $T$ be the set of all final and transient states. We show that $T$ satisfies the conditions of Theorem~\ref{thm-area}. 

For every final state $t \in T$ we have that $t \models \Zero$, and hence $\vec{v}[t] = \vec{\kappa}$. Now let $t \in T$ be a transient state. Since every state of $\As{t}$ is transient or final, we have that $\As{t} \subseteq T$. Furthermore,
\[
    t \models \ex{a,r_i} \wedge \Suc_i \wedge \Interval \wedge \Eq_i
\]   
for precisely one $i \in \{0,\ldots,4\}$. Since $t \models \Interval$, we obtain $\vec{v}[t]_1 \leq \vec{\kappa}_1$. Now consider the formula
    \[
        \Eq_i \ \equiv\ \opf^{2}_{=q} S^2(r_i) \ \wedge \ 
        \opf^{2}_{=q} ((S^2(r_i) \wedge \neg b) \vee (S^3(r_i) \wedge b))
    \] 
The first conjunct says that the probability of all runs initiated in~$t$ satisfying the path formula $\opf^{2} S^2(r_i)$ is equal to~$q$. By inspecting the structure of transient states enforced by $\Suc_i$ (see Fig.~\ref{fig-transient}), we obtain 
\begin{equation*}
    q \ = \ 1 - \vec{v}[t]_1 + \sum_{u \in \As{t}} P(t,u)\cdot \vec{v}[u]_1 \,.
\end{equation*}
The second conjunct of $\Eq_i$ says that the probability of satisfying $\opf^{2} ((S^2(r_i) \wedge \neg b) \vee (S^3(r_i) \wedge b))$ in $t$ is equal to~$q$. Thus, we obtain
\begin{equation*}
        q \ = \  1 - \vec{v}[t]_1 - \vec{v}[t]_2 +  \sum_{u \in \As{t}} P(t,u) \cdot (\vec{v}[u]_1 + \vec{v}[u]_2) \,.
\end{equation*}

Since $s \models \psi[c,d]$, we have that $s$ is a transient state satisfying $\vec{v}[s] = \sigma^n(\vec{\kappa})$. By applying Theorem~\ref{thm-area}, we obtain that if $t$ is a transient state satisfying $\vec{v}[t] = \sigma^k(\vec{\kappa})$ where $1 \leq k \leq n$, then every $u \in \As{t}$ satisfies  $\vec{v}[u] = \tau(\sigma^k(\vec{\kappa})) = \sigma^{k-1}(\vec{\kappa})$. Consequently, for every $i \in \{0,\ldots,n\}$, there must be a state of $S$ with characteristic vector $\sigma^i(\vec{\kappa})$. These states must be pairwise different because $\sigma^i(\vec{\kappa}) \neq \sigma^j(\vec{\kappa})$ for $i \neq j$ (see Lemma~\ref{lem-tausigma}).

To prove the second claim of Theorem~\ref{thm-paramaterized}, realize that $\psi[c,d]$ has a model with states $t_0,\ldots,t_n$, $b_0,\ldots,b_{n-1}$, $c_0,\ldots,c_{n-1}$ where
    \begin{itemize}
        \item for all $1 \leq i \leq n$, $t_i \models \ex{a,r_j}$ where $j = (n{-}i) \textrm{ mod } 5$;
        \item $t_0 \models \ex{h,a,r(j)}$  where $j = n \textrm{ mod } 5$;
        \item for all $1 \leq i \leq n$, we have that $b_{i-1} \models \ex{h,b,S^2(r_j))}$ and $c_{i-1} \models \ex{h,c,S^2(r_j)}$  where  $j = (n{-}i) \textrm{ mod } 5$;
        \item for all $0 \leq i \leq n{-}1$, we have that $P(b_i,b_i) = P(c_i,c_i) =1$;
        \item for all $1 \leq i \leq n$, we have that $P(t_i,t_{i-1}) = \sigma^i(\vec{\kappa})_1$,
          $P(t_i,b_{i-1}) = \sigma^i(\vec{\kappa})_2$, and  $P(t_i,c_{i-1}) = 1 {-} \sigma^i(\vec{\kappa})_1 {-} \sigma^i(\vec{\kappa})_2$;
        \item $P(t_0,t_0) = 1$.
    \end{itemize}  
It is easy to check that $t_n \models \psi[c,d]$.

\section{Simulating Minsky machines with one counter}
\label{sec-counter}

For the rest of this section, we fix $q$, $I_q$, and $\vec{\kappa}$ in the same way as in Section~\ref{sec-geom}, but we additionally require that $\vec{\kappa}_1 + \vec{\kappa}_2 < q - \frac{1}{2}$. Furthermore, we fix a rational constant  $\gamma$ such that $(1{-}q) \vec{\kappa}_2 < \gamma < \frac{3}{4}q - \frac{5}{4}q + \frac{1}{2}q^2$. For example, we can put $q = \frac{13}{16}$, $\vec{\kappa}= (\frac{17}{64},\frac{1}{32})$, and $\gamma = 0.06$.
These additional assumptions are used in the proof of Theorem~\ref{thm-one-counter}.

Let $\M$ be a non-deterministic one-counter Minsky machine with $m$ instructions, and let $\Labels = \{\ell_i \mid 1\leq i \leq m\}$ be a set of fresh atomic propositions. Let $M = (S,P,v)$ be a Markov chain. We say that a state $t \in S$ \emph{represents a configuration $(i,n)$} of $\M$ iff $t \models \ell_i$, $t \not\models \ell_j$ for all $j \neq i$, and $\vec{v}[t] = \sigma^n(\vec{\kappa})$. Furthermore, we say that a state $s \in S$ \emph{simulates $\M$} if $s$ represents $(1,0)$ and every state $t$ reachable from $s$ satisfies the following condition: If $t$ represents a configuration $C$ of $\M$, then at least one immediate successor of $t$ represents a successor configuration of~$C$. Furthermore, for every immediate successor $t'$ of $t$ that does \emph{not} represent a successor of $C$ we have that $t' \models \opg_{=1} \bigwedge_{i=1}^m \neg\ell_i$.

Let $s \in S$ be a state simulating $\M$, and let $s_0,s_1,\ldots$ be a run such that $s_0 = s$ and every $s_i$ represents a configuration $C_i$ of $\M$. Then $C_0,C_1,\ldots$ is a computation of $\M$ \emph{covered by $s$}. Note that $s$ covers at least one but not necessarily all computations of~$\M$.
In this section, we prove the following theorem:

\begin{theorem}
\label{thm-one-counter}
Let $\M$ be a non-deterministic one-counter Minsky machine. Then there is an effectively constructible PCTL formula $\psi$ satisfying the following conditions:
\begin{itemize}
    \item[(A)] For every Markov chain $M$ and every state $s$ of $M$, we have that if $s \models \psi$, then $s$ simulates $\M$.
    \item[(B)] For every computation $\omega$ of $\M$, there exists a Markov chain $M$ and a state $s$ of $M$ such that $s \models \psi$ and $s$ covers $\omega$. 
    Furthermore, if $\omega$ is periodic, then $M$ has finitely many states.
\end{itemize}
\end{theorem}

\subsection{Constructing $\psi$}
\label{sec-one-counter-psi}

For the rest of this section, we fix a non-deterministic one-counter Minsky machine 
$\M \ \equiv\ 1: \Ins_1; \ \cdots  \ m: \Ins_m$. We show how to construct the formula $\psi$ of Theorem~\ref{thm-one-counter}.

Recall that the counter values $0,1,2,\ldots$ are represented by characteristic vectors $\vec{\kappa}$, $\sigma(\vec{\kappa})$, $\sigma^2(\vec{\kappa}),\ldots$. Hence, decrementing and incrementing the counter corresponds to performing the functions $\tau$ and $\sigma$, respectively. Testing the counter for zero is no issue because both components of $\vec{\kappa}$ are rational and can be directly used as probability constraints. The function $\tau$ is implemented in the same way as in the parameterized formula $\psi(x,y)$ constructed in Section~\ref{sec-large}. Consequently, some subformulae of $\psi(x,y)$ are re-used in~$\psi$, possibly after small adjustments. The main challenge is to implement the function~$\sigma$ and orchestrate everything so that Theorem~\ref{thm-area} is still applicable.

The set of propositions used in $\psi$ is $\A = A \cup \Labels \cup \{d,e\}$ where $A$ is the set of propositions occurring in $\psi(x,y)$ (see Section~\ref{sec-large}), and $d,e$ are fresh propositions used for implementing the function $\sigma$. For every $B \subseteq \A$, we use $\Ex{B}$ to denote the formula  $\bigwedge_{p \in B} p \wedge \bigwedge_{p \in \A \smallsetminus B} \neg p$. The intuitive meaning of $\Ex{B}$ is the same as of $\ex{B}$ defined in Section~\ref{sec-large-formula}. The only difference is that the set $A$ is replaced with the richer set~$\A$. 

The structure of $\psi$ closely resembles the structure of $\psi(x,y)$, and the overall intuition behind the subformulae stays essentially the same. We use the same identifiers for denoting adjusted versions of subformulae defined in Section~\ref{sec-large-formula}. If a subformula does not require any adjustment \emph{except for replacing every occurrence of $\ex{\cdot}$ with $\Ex{\cdot}$}, then we re-use this formula and write its identifier in \textit{\textbf{boldface}}. For example, $\bFsuc_i$ denotes the formula
\[
  \opx_{=1} (\Ex{h,a,S(r_i)} \vee \Ex{h,b,S^2(r_i)} \vee \Ex{h,c,S^2(r_i)})\,.   
\]

We put 
\begin{eqnarray*}
     \psi & \equiv &  \Init \ \wedge \ \opg_{=1} \Invariant   
\end{eqnarray*}
where
\begin{eqnarray*}
    \Init & \equiv & \Ex{a,r_0,\ell_1} \ \wedge \ \bZero\\
    \Invariant & \equiv & \bFin \vee \Transient \vee \bFree
\end{eqnarray*}
Note that if $s \models \psi$, then every state $t$ reachable form $s$ can again be classified as either \emph{final}, \emph{transient}, or \emph{free}, depending on whether $t$~satisfies $\bFin$, $\Transient$, or $\bFree$, respectively.  

%
%
Furthermore, we put
\[
    \Transient \ \equiv \ \bTrans \vee \CTrans \vee \LTrans \,.
\]   
As we shall see, every transient state $t$ satisfies precisely one of the formulae $\Ex{a,r_i}$, $c {\wedge} r_i {\wedge} \neg h$, or $\Ex{x,r_i,\ell}$, where $i \in \{0,\ldots,4\}$, \mbox{$\ell \in \Labels$}, and $x \in \{a,b\}$. The formulae $\bTrans$, $\CTrans$, and $\LTrans$ define the properties of~$t$ in the three respective cases. We define
\begin{eqnarray*}
    \CTrans &\equiv & \bigvee_{i \in \{0,\ldots,4\}} c \wedge r_i \wedge \neg h \wedge \Csuc_i \wedge \bInterval \wedge \bEq_i   
\end{eqnarray*}
where
\[
    \Csuc_i \equiv  \opx_{=1} \big( \Ex{a,S(r_i)} \vee \Ex{h,b,S^2(r_i)}
     \vee \ \Ex{h,c,S^2(r_i)}  \vee  \Ex{h,c,S^2(r_i),d} \big)  
\]
Furthermore, we define 
\begin{eqnarray*}   
    \LTrans & \equiv &\bigvee_{i \in \{0,\ldots,4\}} \bigvee_{\ell \in \Labels} \bigvee_{x \in \{a,b\}} 
    \  (\Ex{x,r_i,\ell} \wedge \Step_{i,\ell})
\end{eqnarray*}
where the formula $\Step_{i,\ell}$ is constructed as follows. Let $\Ins_j$ be the instruction associated with $\ell$, i.e., $\ell = \ell_j$. We distinguish two cases.
\bigskip

\noindent 
(A)  If $\Ins_j \equiv \textit{if } c{=}0 \textit{ then goto } L \textit{ else dec } c; \textit{ goto } L'$, then
    \begin{eqnarray*}
        \Step_{i,\ell} & \equiv & \hspace*{2ex} \left(\bZero \Rightarrow (\Zsuc_{i,\ell} \wedge \opx_{=1} (b \Rightarrow \bZero))\right)\\
        && \wedge \left( \neg\bZero \Rightarrow (\Psuc_{i,\ell} \wedge \bInterval \wedge \bEq_i ) \right)
    \end{eqnarray*}
    where 
    \begin{eqnarray*}
        \Zsuc_{i,\ell} & \equiv & \bigvee_{\ell' \in L} \opx_{=1} \big(\Ex{h,a,S(r_i)} \vee \Ex{h,c,S^2(r_i)} \vee \Ex{h,c,S^2(r_i),e}\\[-3ex]       
        && \hspace*{9ex} \vee~  \Ex{b,S^2(r_i),\ell'} \big) \ \wedge \ 
            \opx_{=1{-}q} \Ex{h,c,S^2(r_i),e}\\
%
        \Psuc_{i,\ell} & \equiv & \bigvee_{\ell' \in L'} \opx_{=1} \big(\Ex{h,b,S^2(r_i)} \vee \Ex{h,c,S^2(r_i)} \vee \Ex{h,c,S^2(r_i),e}\\[-3ex]
        &&  \hspace*{9ex} \vee~  \Ex{a,S(r_i),\ell'} \big) \ \wedge \ \opx_{=1{-}q} \Ex{h,c,S^2(r_i),e}  
    \end{eqnarray*}
\noindent
(B) If $\Ins_j \equiv \textit{inc } c; \textit{ goto } L$, then
\begin{eqnarray*}
    \Step_{i,\ell} & \equiv & 
    \hspace*{2ex} (\bZero \Rightarrow \IZsuc_{i,\ell})\\
    && \wedge \ (\neg\bZero \Rightarrow (\IPsuc_{i,\ell} \wedge \bInterval \wedge \bEq_i)\\
    && \wedge \ \opf^2_{=1{-}q} (a\wedge S^3(r_i))\\
    && \wedge \ \opf^2_{=\gamma} \big((b \wedge S^4(r_i) ) \vee d\big)\\
    && \wedge \ \opf^2_{=\gamma} \big((c\wedge S^4(r_i) \wedge e) \vee d \big)
\end{eqnarray*}
where $\gamma$ is the constant fixed at the beginning of Section~\ref{sec-counter} and
\begin{eqnarray*}
    \IZsuc_{i,\ell} & \equiv & \bigvee_{\ell' \in L} \opx_{=1} \big(\Ex{h,a,S(r_i)}  \vee
    \Ex{c,S^2(r_i)} \vee \Ex{c,S^2(r_i),e}\\[-3ex]
    && \hspace*{9ex}  \vee~  \Ex{b,S^2(r_i),\ell'} \big) \ \wedge \ \opx_{=1{-}q} \Ex{c,S^2(r_i),e}\\
%
    \IPsuc_{i,\ell} & \equiv & \bigvee_{\ell' \in L} \opx_{=1} \big(\Ex{a,S(r_i)}  \vee
    \Ex{c,S^2(r_i)} \vee \Ex{c,S^2(r_i),e}\\[-3ex]
    && \hspace*{9ex}  \vee~  \Ex{b,S^2(r_i),\ell'} \big) \ \wedge \ \opx_{=1{-}q} \Ex{c,S^2(r_i),e}
\end{eqnarray*}
This completes the construction of $\psi$.

Intuitively, if $t \models \Ex{x,r_i,\ell_j}$ and $\vec{v}[t]$ encodes the current counter value, then $\Step_{i,\ell_j}$ enforces the simulation of $\Ins_j$. More concretely,  
\begin{itemize}
    \item if $\Ins_j \equiv \textit{if } c{=}0 \textit{ then goto } L \textit{ else dec } c; \textit{ goto } L'$, then 
    \begin{itemize}
        \item if $\vec{v}[t] = \vec{\kappa}$ encodes zero, then each $t' \in \Bs{t}$ satisfies some $\ell' \in L$ and $\vec{v}[t']= \vec{\kappa}$. The simulation continues in the states of~$\Bs{t}$.
        \item if $\vec{v}[t]$ encodes a positive value, then each $t' \in \As{t}$ satisfies some $\ell' \in L'$ and $\vec{v}[t'] = \tau(\vec{v}[t])$. The counter is decremented, and the simulation continues in the states of~$\As{t}$.
    \end{itemize}
    \item If $\Ins_j \equiv \textit{inc } c; \textit{ goto } L$, then each $t' \in \Bs{t}$ satisfies some $\ell' \in L$. Furthermore, $\vec{v}[t'] = \sigma(\vec{v}[t])$ for all $t' \in \Bs{t} \cup \Cs{t}$ (this part is tricky). The counter is incremented and the simulation continues in the states of $\Bs{t}$. 
\end{itemize}
Furthermore, all transient states~$t$ where $\vec{v}[t] \neq \vec{\kappa}$ satisfy the formula $\bInterval \wedge \bEq_i$ (for an appropriate $i$) so that the conditions of Theorem~\ref{thm-area} are fulfilled for the set of all transient and final states.

\tikzstyle{state}=[draw,thick,minimum width=12mm,minimum height=12mm,rounded corners,text centered]
\begin{figure*}[t]\centering
\begin{tikzpicture}[x=1.8cm, y=1.7cm,scale=0.80,font=\small]
\node (t) at (-1,0) [state, draw=none] {\LARGE$\bullet$};
\node [above of = t,node distance=3ex] (pt) {$\Ex{x,r_i,\ell_j}$};
\node [left of = t, node distance=2ex] (lt) {$t$};
\begin{scope}[shift={(6,0.5)}]
    \node (u) at (-2,-1.5) [state] {\LARGE$\bullet\ \cdots \ \bullet$};
    \node [above of = u, node distance=2.5ex] (lu) {$\As{t}$};
    \node [below of = u, xshift=-4ex, node distance=2ex] (luu) {$t_a'$};
    \node [below of = u, node distance=6.5ex] (pu) {$\Ex{a,S(r_i)}$};
    \node (z1) at (-5.2,-2.7) [state] {};
    \node [below of = z1] (pz1) {$a,S^2(r_i)$};
    \node (z2) at (-4,-2.7) [state] {};
    \node [below of = z2] (pz2) {$\Ex{h,b,S^3(r_i)}$};
    \node (z3) at (-2.8,-2.7) [state] {};
    \node [below of = z3] (pz3) {$\Ex{h,c,S^3(r_i)}$};
    \draw[-stealth,rounded corners,thick] (t.center) -| node[left, near end] {$\vec{v}[t]_1$} (u);
    \draw[-stealth,rounded corners,thick] ([xshift=-4ex]u.center) -| node[left, near end] {$\vec{v}[t_a']_1$} (z1);
    \draw[-stealth,rounded corners,thick] ([xshift=-4ex]u.center) -| node[left, near end] {$\vec{v}[t_a']_2$} (z2);
    \draw[-stealth,rounded corners,thick] ([xshift=-4ex]u.center) -| node[left, near end] {} (z3);
\end{scope}
\begin{scope}[shift={(12,0.5)}]
    \node (u) at (-2,-1.5) [state] {\LARGE$\bullet\ \cdots \ \bullet$};
    \node [above of = u, node distance=2.5ex] (lu) {$\Bs{t}$};
    \node [below of = u, xshift=-3.5ex, node distance=2ex] (luu) {$t_b'$};
    \node [below of = u, node distance=6.5ex] (pu) {$\Ex{b,S^2(r_i),\ell'}$};
    \node (z1) at (-6.4,-2.7) [state] {};
    \node [below of = z1] (pz1) {$a,S^3(r_i)$};
    \node (z2) at (-5.2,-2.7) [state] {};
    \node [below of = z2] (pz2) {$b,S^4(r_i)$};
    \node (z3) at (-4,-2.7) [state] {};
    \node [below of = z3] (pz3) {$c,S^4(r_i),\neg e$};
    \node (z4) at (-2.8,-2.7) [state] {};
    \node [below of = z4] (pz4) {$c,S^4(r_i),e$};
    \draw[-stealth,rounded corners,thick] (t.center) -| node[left, near end] {$\vec{v}[t]_2$} (u);
    \draw[-stealth,rounded corners,thick] ([xshift=-4ex]u.center) -| node[left, near end] {$\vec{v}[t_b']_1$} (z1);
    \draw[-stealth,rounded corners,thick] ([xshift=-4ex]u.center) -| node[left, near end] {$\vec{v}[t_b']_2$} (z2);
    \draw[-stealth,rounded corners,thick] ([xshift=-4ex]u.center) -| node[left, near end] {} (z3);   
    \draw[-stealth,rounded corners,thick] ([xshift=-4ex]u.center) -| node[left, near end] {$1{-}q$} (z4);
\end{scope}
\begin{scope}[shift={(6,-3)}]
    \node (u) at (-2,-1.5) [state] {\LARGE$\bullet\ \cdots \ \bullet$};
    \node [above of = u, node distance=2ex] (lu) {};
    \node [below of = u, xshift=-3.5ex, node distance=2ex] (luu) {$t_c'$};
    \node [below of = u, node distance=6.5ex] (pu) {$\Ex{c,S^2(r_i)}$};
    \node (z1) at (-6.4,-2.7) [state] {};
    \node [below of = z1] (pz1) {$\Ex{a,S^3(r_i)}$};
    \node (z2) at (-5.2,-2.7) [state] {};
    \node [below of = z2] (pz2) {$\Ex{h,b,S^4(r_i)}$};
    \node (z3) at (-4,-2.7) [state] {};
    \node [below of = z3] (pz3) {$\Ex{h,c,S^4(r_i)}$};
    \node (z4) at (-2.8,-2.7) [state] {};
    \node [below of = z4] (pz4) {$\Ex{h,c,S^4(r_i),d}$};
    \draw[-stealth,rounded corners,thick] (t.center) -- ++(.6,0) -- ++(0,-3.5) -|  (u);
    \draw[-stealth,rounded corners,thick] ([xshift=-4ex]u.center) -| node[left, near end] {$\vec{v}[t_c']_1$} (z1);
    \draw[-stealth,rounded corners,thick] ([xshift=-4ex]u.center) -| node[left, near end] {$\vec{v}[t_c']_2$} (z2);
    \draw[-stealth,rounded corners,thick] ([xshift=-4ex]u.center) -| node[left, near end] {} (z3);   
    \draw[-stealth,rounded corners,thick] ([xshift=-4ex]u.center) -|  (z4);
\end{scope}
\begin{scope}[shift={(12,-3)}]
    \node (u) at (-2,-1.5) [state] {\LARGE$\bullet\ \cdots \ \bullet$};
    \node [above of = u, node distance=2ex] (lu) {};
    \node [below of = u, xshift=-3ex, node distance=2ex] (luu) {$t_c''$};
    \node [below of = u, node distance=6.5ex] (pu) {$\Ex{c,S^2(r_i),e}$};
    \node (z1) at (-6.4,-2.7) [state] {};
    \node [below of = z1] (pz1) {$\Ex{a,S^3(r_i)}$};
    \node (z2) at (-5.2,-2.7) [state] {};
    \node [below of = z2] (pz2) {$\Ex{h,b,S^4(r_i)}$};
    \node (z3) at (-4,-2.7) [state] {};
    \node [below of = z3] (pz3) {$\Ex{h,c,S^4(r_i)}$};
    \node (z4) at (-2.8,-2.7) [state] {};
    \node [below of = z4] (pz4) {$\Ex{h,c,S^4(r_i),d}$};
    \draw[-stealth,rounded corners,thick] (t.center) -- ++(.6,0) -- ++(0,-3.5) -| node[left, near end] {$1-q$}  (u);
    \draw[-stealth,rounded corners,thick] ([xshift=-4ex]u.center) -| node[left, near end] {$\vec{v}[t_c'']_1$} (z1);
    \draw[-stealth,rounded corners,thick] ([xshift=-4ex]u.center) -| node[left, near end] {$\vec{v}[t_c'']_2$} (z2);
    \draw[-stealth,rounded corners,thick] ([xshift=-4ex]u.center) -| node[left, near end] {} (z3);   
    \draw[-stealth,rounded corners,thick] ([xshift=-4ex]u.center) -|  (z4);
\end{scope}

\end{tikzpicture}
\caption{The structure of transient states satisfying $\Ex{x,r_i,\ell_j}$ where $x \in \{a,b\}$ and $\Ins_j \equiv \textit{inc } c; \textit{ goto } u$}
\label{fig-transient-oc}
\end{figure*}
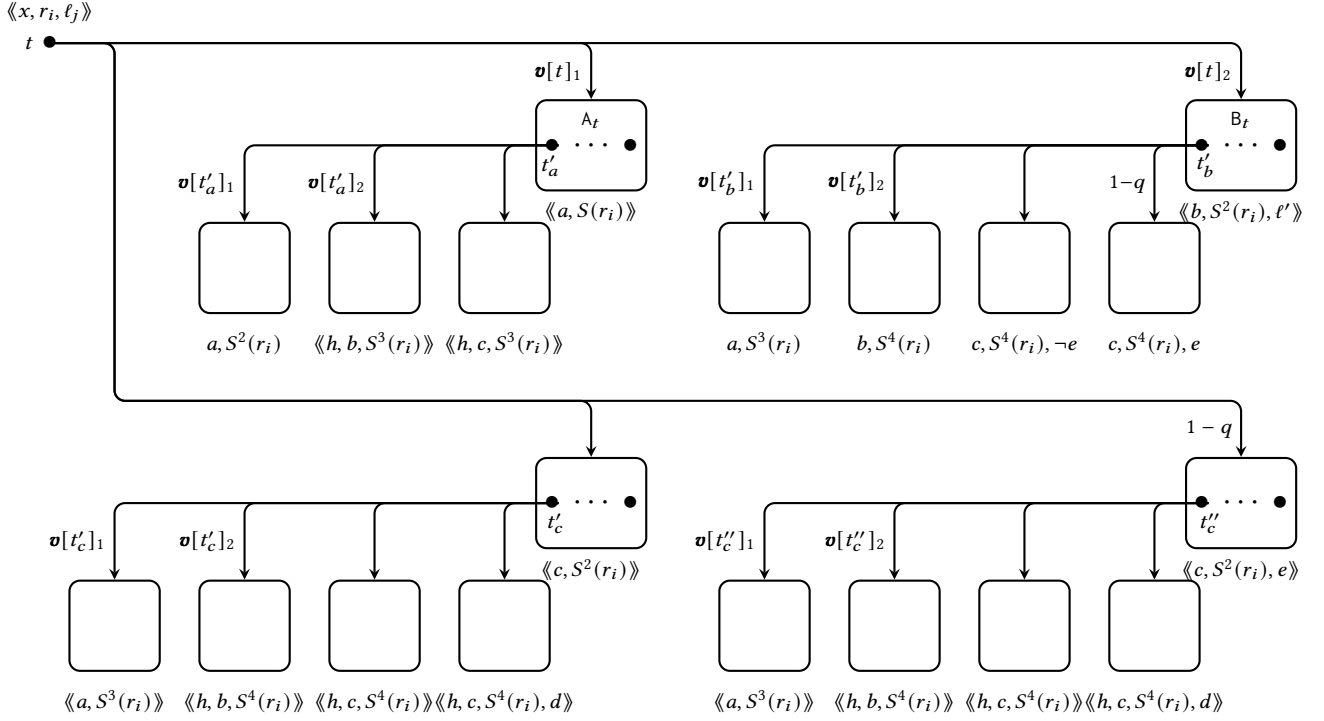

\subsection{A Proof of Theorem~\ref{thm-one-counter}}

Let $\M \ \equiv\ 1: \Ins_1; \ \cdots  \ m: \Ins_m$ be a non-deterministic one-counter Minsky machine, and let $\psi$ be the formula constructed for $\M$ in Section~\ref{sec-one-counter-psi}. The two claims of Theorem~\ref{thm-one-counter} are proven separately in the following subsections.




\subsubsection{A Proof of Theorem~\ref{thm-one-counter}~(A)}
\label{sec-prop-one-counter-unbounded}

Let $M =(S,P,v)$ be Markov chain and $s \in S$ such that $s \models \psi$.  Let $T$ be the set of all transient states and all final states of $M$. One can easily verify that the conditions of Theorem~\ref{thm-area} are satisfied for $T$ by inspecting the structure of transient states similarly as in Section~\ref{sec-proof-main-param} (in particular, the formula $\bEq_i$ still implies Equalities~\eqref{eq-convex-one} and~\eqref{eq-convex-two}). Thus, by applying Theorem~\ref{thm-area}, we obtain 
\begin{itemize}
    \item[I.] For every $t \in T$, we have that $\vec{v}[t] \in \Area(\vec{\kappa})$.
    \item[II.] For every $t \in T$ where $\vec{v}[t] \neq \vec{\kappa}$ and every $t' \in \As{t}$, we have that $\vec{v}[t']  = \tau(\vec{v}[t])$.
\end{itemize}

We show that $s$ simulates $\M$. Since $s \models \Init$, we have that $s$~represents the configuration $(1,0)$. Let $t$ be a state reachable from $s$ such that $t$ 
represents a configuration $(j,n)$ of~$\M$, i.e., $t \models \ell_j$ and $\vec{v}[t] = \sigma^n(\vec{\kappa})$. Then $t \models \LTrans$ and hence also $t \models \Step_{i,\ell_j}$. We distinguish two cases.
\smallskip

\noindent
(A) $\Ins_j \equiv \textit{if } c{=}0 \textit{ then goto } L \textit{ else dec } c; \textit{ goto } L'$.\quad Then, 
\begin{eqnarray*}
    \Step_{i,\ell} & \equiv & \hspace*{2ex} \left(\bZero \Rightarrow (\Zsuc_{i,\ell} \wedge \opx_{=1} (b \Rightarrow \bZero))\right)\\
    && \wedge \left( \neg\bZero \Rightarrow (\Psuc_{i,\ell} \wedge \bInterval \wedge \bEq_i ) \right)
\end{eqnarray*}
and there are two subcases.
\begin{itemize}
    \item $n=0$.  Since $\vec{v}[t] = \sigma^n(\vec{\kappa})$, we obtain $t \models \bZero$. Hence, $t \models \Zsuc_{i,\ell} \wedge \opx_{=1} (b \Rightarrow \bZero)$. This implies that every $t' \in \Bs{t}$ satisfies $t' \models \Ex{b,S^2(r_i),\ell'} \wedge \bZero$ for some $\ell' \in L$. Hence, $t'$ represents a successor configuration of $(j,0)$. 
    Since $\vec{v}[t]_2 = \vec{\kappa}_2 > 0$, we have $\Bs{t} \neq \emptyset$, and hence at least one such $t'$ exists. Furthermore, all states of $\As{t} \cup \Cs{t}$ are free and hence they satisfy the formula $\opg_{=1} \bigwedge_{i=1}^m \neg\ell_i$.
    \item $n > 0$. Since $\vec{v}[t] = \sigma^n(\vec{\kappa})$, we obtain $t \models \neg\bZero$ and hence $t \models \Psuc_{i,\ell} \wedge \bInterval \wedge \bEq_i$. The formula $\Psuc_{i,\ell}$ ensures that every $t' \in \As{t}$  satisfies $t' \models \Ex{a,S(r_i),\ell'}$ for some \mbox{$\ell' \in L'$}. Furthermore, we have that $\vec{v}[t'] = \tau(\sigma^n(\vec{\kappa})) =  \sigma^{n-1}(\vec{\kappa})$ by Observation~II.~above. Hence, $t'$ represents a successor configuration of $(j,n)$. Since $\vec{v}[t]_1 = \sigma^n(\vec{\kappa})_1 > 0$, at least one such $t'$ must exist. Note that all states of $\Bs{t} \cup \Cs{t}$ satisfy the formula $\opg_{=1} \bigwedge_{i=1}^m \neg\ell_i$.
\end{itemize}

\noindent
(B) $\Ins_j \equiv \textit{inc } c; \textit{ goto } L$.\quad Then, 
\begin{eqnarray*}
    \Step_{i,\ell} & \equiv & 
    \hspace*{2ex} (\bZero \Rightarrow \IZsuc_{i,\ell})\\
    && \wedge \ (\neg\bZero \Rightarrow (\IPsuc_{i,\ell} \wedge \bInterval \wedge \bEq_i)\\
    && \wedge \ \opf^2_{=1{-}q} (a\wedge S^3(r_i))\\
    && \wedge \ \opf^2_{=\gamma} \big((b \wedge S^4(r_i) ) \vee d\big)\\
    && \wedge \ \opf^2_{=\gamma} \big((c\wedge S^4(r_i) \wedge e) \vee d \big)
\end{eqnarray*}
 
The structure of immediate successors of $t$ enforced by this formula is shown in Fig.~\ref{fig-transient-oc}. The formulae $\IZsuc_{i,\ell}$ and $\IPsuc_{i,\ell}$ ensure that for every $t' \in \Bs{t}$, we have that $t' \models \Ex{b,S^2(r_i),\ell'}$ where $\ell' \in L$. 
Observe that $\Bs{t} \neq \emptyset$ and all states of $\As{t} \cup \Cs{t}$ satisfy the formula $\opg_{=1} \bigwedge_{i=1}^m \neg\ell_i$. Hence, it remains to show that $\vec{v}[t'] = \sigma(\vec{v}[t])$ for all $t' \in \Bs{t}$. We prove a stronger result saying that
$\vec{v}[t'] = \sigma(\vec{v}[t])$ for all $t' \in \Bs{t} \cup \Cs{t}$.

Due to Observation~I above and Lemma~\ref{lem-gproperties}.A, it suffices to prove that $\sigma(\vec{v}[t])$ is a positive convex combination of the vectors in $\{\vec{v}[t'] \mid t' \in \Bs{t} \cup \Cs{t}\}$. For a path formula $\Phi$, we use $R[\Phi]$ to denote the set of all $w \in \run(t)$ such that $w \models \Phi$.

Since $t \models \opf^2_{=1{-}q} (a\wedge S^3(r_i))$, we have $\m_t(R[\opf^2(a\wedge S^3(r_i))]) = 1{-}q$.
By inspecting the structure of immediate successors of~$t$ (see Fig.~\ref{fig-transient-oc}), we obtain
\begin{equation}
   1{-}q \ = \ \m_t(R[\opf^2(a\wedge S^3(r_i))]) \ = \ \sum_{t' \in \Bs{t} \cup \Cs{t}} P(t,t') \cdot \vec{v}[t']_1\,.
\label{eq-sigma-one}   
\end{equation}
Observe that $\sum_{t' \in \Bs{t} \cup \Cs{t}} P(t,t') = 1 {-} \vec{v}[t]_1$. Hence, \eqref{eq-sigma-one} implies
\begin{equation}
   \frac{1{-}q}{1 {-} \vec{v}[t]_1} \ =  \ \sum_{t' \in \Bs{t} \cup \Cs{t}} \frac{P(t,t')}{1 {-} \vec{v}[t]_1} \cdot \vec{v}[t']_1 \,.
\label{eq-sigma-first}
\end{equation}
Note that the left-hand side of~\eqref{eq-sigma-first} is equal to $\sigma(\vec{v}[t])_1$. 

Furhermore, $t \models \opf^2_{=\gamma} \big((b \wedge S^4(r_i) ) \vee d\big)$. By inspecting the structure of immediate successors of~$t$ (see Fig.~\ref{fig-transient-oc}), we obtain\footnote{Here, by writing $C = A \uplus B$ we mean that $C = A \cup B$ and $A \cap B = \emptyset$.}
\begin{eqnarray}
   \gamma & = & \m_t(R[\opf^2((b \wedge S^4(r_i) ) \vee d)]) \nonumber \\
       & = & \m_t(R[\opf^2(b \wedge S^4(r_i))] \uplus R[\opf^2(d)]) \nonumber \\
       & = & \m_t(R[\opf^2(b \wedge S^4(r_i))]) \  \ + \ \  \m_t(R[\opf^2(d)]) \nonumber \\
       & = & \sum_{t' \in \Bs{t} \cup \Cs{t}} P(t,t') \cdot \vec{v}[t']_2 \ \ + \ \ \m_t(R[\opf^2(d)])
\label{eq-sigma-F1}
\end{eqnarray}

Similarly, $t \models \opf^2_{=\gamma} \big((c\wedge S^4(r_i) \wedge e) \vee d \big)$ implies
\begin{eqnarray}
    \gamma & = & \m_t(R[\opf^2((c\wedge S^4(r_i) \wedge e) \vee d)]) \nonumber \\
        & = & \m_t(R[\opf^2(c\wedge S^4(r_i) \wedge e)] \uplus R[\opf^2(d)]) \nonumber \\
        & = & \m_t(R[\opf^2(c \wedge S^4(r_i) \wedge e)]) \  \ + \ \  \m_t(R[\opf^2(d)]) \nonumber \\
        & = & \sum_{t' \in \Bs{t}} P(t,t') \cdot (1{-}q) \ \ + \ \ \m_t(R[\opf^2(d)]) \nonumber \\
        & = & \vec{v}[t]_2 \cdot (1{-}q) \ \ + \ \ \m_t(R[\opf^2(d)])
 \label{eq-sigma-F2}
\end{eqnarray}

Since the right-hand sides of \eqref{eq-sigma-F1} and \eqref{eq-sigma-F2} are equal, we have that
\begin{equation*}
    \vec{v}[t]_2 \cdot (1{-}q) = \sum_{t' \in \Bs{t} \cup \Cs{t}} P(t,t') \cdot \vec{v}[t']_2
\end{equation*}
Hence,
\begin{equation}
    \frac{\vec{v}[t]_2 \cdot (1{-}q)}{1 {-} \vec{v}[t]_1} = \sum_{t' \in \Bs{t} \cup \Cs{t}} \frac{P(t,t')}{1 {-} \vec{v}[t]_1} \cdot \vec{v}[t']_2
\label{eq-sigma-second}
\end{equation}
Observe that the left-hand side of~\eqref{eq-sigma-second} is equal to $\sigma(\vec{v}[t])_2$. By combining~\eqref{eq-sigma-first} and~\eqref{eq-sigma-second}, we finally obtain
\[
    \sigma(\vec{v}[t]) \ = \ \sum_{t' \in \Bs{t} \cup \Cs{t}} \frac{P(t,t')}{1 {-} \vec{v}[t]_1} \cdot \vec{v}[t']\,.
\] 

\subsubsection{A Proof of Theorem~\ref{thm-one-counter}~(B)}
\label{sec-proof-prop-one-counter-bounded}

For the rest of this proof, we fix a computation $\omega = C_0,C_1,\ldots$ of $\M$. Furthermore, if $\omega$ is periodic, we fix $\alpha,\beta$ such that $\alpha < \beta$ and the infinite sequences $C_{\alpha-1},C_{\alpha},C_{\alpha+1},\ldots$ and $C_{\beta-1},C_{\beta},C_{\beta+1},\ldots$ are the same. If $\omega$ is not periodic, we put $\beta = \infty$ and leave $\alpha$ undefined.

Observe that if $\omega$ is periodic, then our choice of $\alpha$ and $\beta$ ensures that the infinite sequences $C_{\alpha},C_{\alpha+1},\ldots$ and $C_\beta,C_{\beta+1},\ldots$ are also the same, and the compuational steps $C_{\alpha-1} \mapsto C_{\alpha}$ and  $C_{\beta-1} \mapsto C_{\beta}$ are generated by the same instruction of $\M$. As we shall see, this ensures that $C_\alpha$ and $C_\beta$ are represented by the same state in the constructed model of $\psi$. 

We show that there exist a Markov chain $M = (S,P,v)$ and a state $s \in S$ such that $s \models \psi$ and $s$ covers $\omega$. Furthermore, if $\beta < \infty$, then $S$ is finite.

Every state of $S$ is a triple of the form $[\iota,\calL,n]$ where $\iota$ is an index ranging over $\{i \in \N \mid 0\leq i < \beta\} \cup \{\star\}$, $\calL \subseteq \A$ is the set of atomic propositions satisfied in the state, and $n \in \N \cup \{\star\}$ is a counter value. The $\star$ symbol is used when the index or the counter value (or both) are not relevant. 

The Markov chain $M$ is the least Markov chain $M'$ such that $s \equiv [0,\{a,r_0,\ell_1\},0]$ is a state of $M'$, and if $t$ is a state of $M'$ then $M'$ contains all immediate successors of $t$ defined by the following rules:
\smallskip

\noindent
\textbf{Rule I.} If $t = [k,\{x,r_i,\ell_j\},n]$ where $x \in \{a,b\}$, $C_{k} = (j,n)$, and $k < \beta$, then the immediate successors of $t$ are determined as follows. First, let $k'$ be either $\alpha$ or $k+1$ depending on whether $k=\beta - 1$ or $k < \beta -1$, respectively.
Furthermore, let $C_{k'} = (j',n')$, and let $t'$ be either $[k',\{a,S(r_i),\ell_{j'}\},n']$ or $[k',\{b,S^2(r_i),\ell_{j'}\},n']$, depending on whether $n' = n-1$ or $n' \geq n$, respectively. 

Now we distinguish four possibilities (the cases when $\Ins_j$ is a Type~II and Type~I instruction are covered by~(A)--(B) and (C)--(D), respectively. In each case, we distinguish between zero and positive counter values represented by the~$n$).
\smallskip

\noindent
\textbf{(A)} $n' = n = 0$. Then,
\begin{eqnarray*}
    P(t,[\star,\{h,a,S(r_i)\},\star]) & = & \vec{\kappa}_1,\\
    P(t,t') & = & \vec{\kappa}_2,\\
    P(t,[\star,\{h,c,S^2(r_i),e\},\star]) & = & 1 {-} q,\\
    P(t,[\star,\{h,c,S^2(r_i)\},\star]) & = & q {-} \vec{\kappa}_1 {-} \vec{\kappa}_2 \,.
\end{eqnarray*}

\noindent
\textbf{(B)} $n' = n-1$. We put
\begin{eqnarray*}
    P(t,t') & = & \sigma^n(\vec{\kappa})_1,\\
    P(t,[\star,\{h,b,S^2(r_i)\},\star]) & = & \sigma^n(\vec{\kappa})_2,\\ 
    P(t,[\star,\{h,c,S^2(r_i),e\},\star]) & = & 1 {-} q,\\
    P(t,[\star,\{h,c,S^2(r_i)\},\star]) & = & q {-} \sigma^n(\vec{\kappa})_1 {-} \sigma^n(\vec{\kappa})_2 \,.
\end{eqnarray*}

\noindent
\textbf{(C)} $n' = 1$ and $n = 0$. Then,
\begin{eqnarray*}
    P(t,[\star,\{h,a,S(r_i)\},\star]) & = & \vec{\kappa}_1,\\
    P(t,t') & = & \vec{\kappa}_2,\\
    P(t,[\star,\{c,S^2(r_i),e\},1]) & = & 1 {-} q,\\
    P(t,[\star,\{c,S^2(r_i)\},1]) & = & q {-} \vec{\kappa}_1 {-} \vec{\kappa}_2 
\end{eqnarray*}

\noindent
\textbf{(D)} $n' = n+1$ and $n > 0$. Then,
\begin{eqnarray*}
    P(t,[\star,\{a,S(r_i)\},n{-}1]) & = & \sigma^n(\vec{\kappa})_1,\\
    P(t,t') & = & \sigma^n(\vec{\kappa})_2,\\
    P(t,[\star,\{c,S^2(r_i),e\},n{+}1]) & = & 1{-}q,\\
    P(t,[\star,\{c,S^2(r_i)\},n{+}1]) & = & q {-} \sigma^n(\vec{\kappa})_1 {-} \sigma^n(\vec{\kappa})_2 
\end{eqnarray*}

\noindent
\textbf{Rule II.} If $t = [\star,\{a,r_i\},n]$ where $n \geq 0$, we have the following:
\begin{eqnarray*}
    \mbox{if $n = 0$, then}\hspace*{1.2em} P(t,[\star,\{h,a,S(r_i)\},\star]) & = & 
       \sigma^n(\vec{\kappa})_1 = \vec{\kappa}_1,\\
    \mbox{if $n > 0$, then}\hspace*{1em} P(t,[\star,\{a,S(r_i)\},n{-}1]) & = & \sigma^n(\vec{\kappa})_1,\\
    P(t,[\star,\{h,b,S^2(r_i)\},\star]) & = & \sigma^n(\vec{\kappa})_2,\\
    P(t,[\star,\{h,c,S^2(r_i)\},\star]) & = & 1 {-} \sigma^n(\vec{\kappa})_1 {-} \sigma^n(\vec{\kappa})_2 
\end{eqnarray*}

\noindent
\textbf{Rule III.} If $t = [\star,\calL,n]$ where $n>0$ and $\calL$ is $\{c,r_i\}$ or $\{c,r_i,e\}$, then
\begin{eqnarray*}
    P(t,[\star,\{a,S(r_i)\},n{-}1))         & = & \sigma^n(\vec{\kappa})_1,\\
    P(t,[\star,\{h,b,S^2(r_i)\},\star)      & = & \sigma^n(\vec{\kappa})_2,\\
    P(t,[\star,\{h,c,S^2(r_i),d\},\star)    & = & p_n,\\
    P(t,[\star,\{h,c,S^2(r_i)\},\star)      & = & 1 - \sigma^n(\vec{\kappa})_1 - \sigma^n(\vec{\kappa})_2 - p_n
\end{eqnarray*}
where 
\[
     p_n = \frac{\gamma - (1{-}q) \sigma^{n-1}(\vec{\kappa})_2}{1 - \sigma^n(\vec{\kappa})_1 - \sigma^n(\vec{\kappa})_2}    
\]
Note that $0 < p_n < 1 - \sigma^n(\vec{\kappa})_1 - \sigma^n(\vec{\kappa})_2$ due to the constraints imposed on $\gamma$ and $\vec{\kappa}$ at the beginning of Section~\ref{sec-counter}. More precisely, the constraint $\gamma > (1{-}q) \vec{\kappa}_2$ ensures that $p_n > 0$, and the constraint $\gamma < \frac{3}{4}q - \frac{5}{4}q + \frac{1}{2}q^2$ ensures that 
$2p_n < 1 - (2 (q{-}\frac{1}{2})+ (1{-}q))$ (recall that $\vec{\kappa}_1 + \vec{\kappa}_1 < q - \frac{1}{2}$ and $\sigma^{m}(\vec{\kappa})$ is strictly less than $\vec{\kappa}$ in both components for every $m \geq 0$). This is more than we need here; the constraints are chosen so that they also satisfy stronger requirements needed in the proof of Theorem~\ref{thm-two-counter}.
\smallskip

\noindent
\textbf{Rule IV.} If $t =[\star,\calL,\star]$, then $P(t,t) = 1$. 
\smallskip

Note that if $\omega$ is periodic, then $M$ has finitely many states. It is easy to check that $s \models \psi$ by verifying that every state of $M$ satisfies either $\bFin$, $\Transient$, or $\bFree$. In particular, every state of the form $[\iota,\{b,r_i,\ell_j\},n]$ where $\Ins_j$ is a Type~I instruction satisfies
$\opf^2_{=\gamma} \big((b \wedge S^4(r_i) ) \vee d\big)$ and
$\opf^2_{=\gamma} \big((c\wedge S^4(r_i) \wedge e) \vee d \big)$. This is where we need the constant $p_n$ of Rule~III. Clearly, $s$ covers the computation~$\omega$ (and no other computation).

\section{Simulating Minsky Machines with Two Counters}
\label{sec-Minsky-simulation}

In this section, we show how to simulate non-deterministic two-counter Minsky machines by  PCTL formulae. Technically, we construct a PCTL formula simulating a synchronized products of two non-deterministic one-counter Minsky machines defined in the next paragraph.

Let 
\begin{eqnarray*}
   \M_1 & \equiv & 1:\Ins_1^1;\cdots m: \Ins_m^1;\\ 
   \M_2 & \equiv & 1:\Ins_1^2;\cdots m: \Ins_m^2;
\end{eqnarray*}
be nondetermnistic one-counter Minsky machines with $m$ instructions and $I = (I_1,I_2)$ a partition of $\{1,\ldots,m\}$, i.e., $I_1 \cup I_2 = \{1,\ldots,m\}$ and $I_1 \cap I_2 = \emptyset$. An \emph{$I$-synchronized product} of $\M_1,\M_2$, denoted by $\M_1 \times_I \M_2$, is an automaton\footnote{$\M_1 \times_I \M_2$ is a non-standard computational model introduced specifically for purposes of this paper. Encoding the computation of $\M_1 \times_I \M_2$ by a PCTL formula is substantially easier than encoding the computation of a two-counter Minsky machine.} operating over two counters in the following way.

A \emph{configuration} of $\M_1 \times_I \M_2$ is a triple $(j,n_1,n_2)$ where $j \in \{1,\ldots,m\}$ and $n_1,n_2 \in \N$ are counter values. For every configuration $(j,n_1,n_2)$ of $\M_1 \times_I \M_2$, the \emph{successor} configurations $(j',n_1',n_2')$ are determined as follows (recall that $\mapsto$ denotes the ``standard'' computational step of a Minsky machine, see Section~\ref{sec-Minsky}):

\begin{itemize}
    \item If $j \in I_1$, then $j',n_1',n_2'$ are integers satisfying $(j,n_1) \mapsto (j',n_1')$ in $\M_1$ and $(j,n_2) \mapsto (j'',n_2')$ in $\M_2$ (for some $j'' \in \{1,\ldots,m\}$).
    \item If $j \in I_2$, then $j',n_1',n_2'$ are integers satisfying $(j,n_1) \mapsto (j'',n_1')$ in $\M_1$ and $(j,n_2) \mapsto (j',n_2')$ in $\M_2$ (for some $j'' \in \{1,\ldots,m\}$).
\end{itemize}

In other words, $\M_1 \times_I \M_2$ simultaneously executes the instructions $\Ins^1_j$ and $\Ins^2_j$ operating on the first and the second counter, and the next $j'$ is determined by either $\Ins^1_j$ or $\Ins^2_j$, depending on whether $j \in I_1$ or $j \in I_2$, respectively.
Note that the counter values $n_1'$ and $n_2'$ are the \emph{same} in every successor configuration of $(j,n_1,n_2)$.

We write $(j,n_1,n_2) \leadsto (j',n_1',n_2')$ when $(j',n_1',n_2')$ is a successor of $(j,n_1,n_2)$. A \emph{computation} of $\M_1 \times_I \M_2$ is an infinite sequence of configurations $\omega \equiv D_0,D_1,\ldots$ such that $D_i \leadsto D_{i+1}$ for all $i \in \N$. The boundedness and the recurrent reachability problems for \mbox{$\M_1 \times_I \M_2$} are defined in the same way as for non-deterministic Minsky machines (see Section~\ref{sec-Minsky}). 

A synchronized product of two one-counter Minsky machines can faithfully simulate a two-counter Minsky machine. Thus, we obtain the following:

\begin{restatable}{proposition}{Minskyproduct}
\label{prop-product}
    The boundedness problem for a synchronized product of two deterministic one-counter Misky machines is $\Sigma_1^0$-hard. The recurrent reachability problem for a synchronized product of two non-deterministic one-counter Misky machines is $\Sigma_1^1$-hard.
\end{restatable}
 

Let $\M_1 \times_I \M_2$ be a synchronized product of two non-deterministic one-counter Misky machines with~$m$ instructions. Recall the set $\A$ defined in Section~\ref{sec-one-counter-psi}. For $k \in \{1,2\}$, let $\A^k = \{p^k \mid p \in \A\}$ be a set of atomic propositions such that $\A^1 \cap \A^2 = \emptyset$.

Let $M = (S,P,v)$ be a Markov chain. For all $t \in S$ and $k \in \{1,2\}$, we use $\vec{v}^k[t] \in [0,1]^2$ to denote the $k$-th characteristic vector of $t$ where
\begin{itemize}
    \item $\vec{v}^k[t]_1$ is the probability of satisfying the path formula $\opx a^k$ in the state~$t$;
    \item $\vec{v}^k[t]_2$ is the probability of satisfying the path formula $\opx b^k$ in the state~$t$.
\end{itemize}
A state $t \in S$ \emph{represents a configuration $(i,n_1,n_2)$} of $\M_1 \times_I \M_2$ iff $t \models \ell^1_i \wedge \ell^2_i$, $t \not\models \ell_j^1 \wedge \ell_j^2$ for all $j \neq i$, $\vec{v}^1[t] = \sigma^{n_1}(\vec{\kappa})$, and $\vec{v}^2[t] = \sigma^{n_2}(\vec{\kappa})$. Furthermore, we say that a state $s \in S$ \emph{simulates $\M_1 \times_I \M_2$} if $s$ represents $(1,0,0)$ and every state $t$ reachable from~$s$ satisfies the following condition: If $t$ represents a configuration $D$ of $\M_1 \times_I \M_2$, then at least one immediate successor of $t$ represents a successor configuration of~$D$. Furthermore, for every immediate successor $t'$ of $t$ that does \emph{not} represent a successor of $D$ we have that $t' \models \opg_{=1} \bigwedge_{i=1}^m \neg(\ell_i^1 \wedge \ell_i^2)$. 

Let $s \in S$ be a state simulating $\M_1 \times_I \M_2$, and let $s_0,s_1,\ldots$ be a run of $M$ such that $s_0 = s$ and every $s_i$ represents a configuration $D_i$ of $\M_1 \times_I \M_2$. Then $D_0,D_1,\ldots$ is a computation of $\M_1 \times_I \M_2$ \emph{covered by $s$}. 

Now, we formulate the main technical result of this paper. 

\begin{theorem}
\label{thm-two-counter}
Let $\M_1 \times_I \M_2$ be a synchronized product of two non-deterministic one-counter Minsky machines. Then there is an effectively constructible PCTL formula $\Psi$ satisfying the following conditions:
\begin{itemize}
    \item[(A)] For every Markov chain $M$ and every state $s$ of $M$, we have that if $s \models \Psi$, then $s$ simulates $\M_1 \times_I \M_2$.
    \item[(B)] For every computation $\omega$ of $\M_1 \times_I \M_2$, there exists a Markov chain $M$ and a state $s$ of $M$ such that $s \models \Psi$ and $s$ covers $\omega$. 
    Furthermore, if $\omega$ is periodic, then $M$ has finitely many states.
\end{itemize}
\end{theorem}

Observe the following:
\begin{itemize}
    \item Let $\M_1 \times_I \M_2$ be a synchronized product of two \emph{deterministic} one-counter Minsky machines. Then, $\M_1 \times_I \M_2$ is bounded iff the only computation $\omega$ of $\M_1 \times_I \M_2$ is periodic. By Theorem~\ref{thm-two-counter}, we obtain that $\M_1 \times_I \M_2$ is bounded iff $\Psi$ is finite-satisfiable.
    \item Let $\M_1 \times_I \M_2$ be a synchronized product of two \emph{non-deterministic} one-counter Minsky machines. Then \mbox{$\M_1 \times_I \M_2$} has a recurrent computation iff the formula 
    \begin{equation*}
        \Psi \ \wedge \ \opg_{=1} \left( (\ell_1^1 \wedge \ell_1^2) \Rightarrow \opf_{>0} (\ell_1^1 \wedge \ell_1^2)\right)
    \end{equation*}
    is (generally) satisfiable.
\end{itemize} 

Thus, we obtain the following corollary to Theorem~\ref{thm-two-counter}:

\begin{corollary}
    The finite satisfiability problem for PCTL is \mbox{$\Sigma_1^0$-hard}, and the general satisfiability problem pro PCTL is \mbox{$\Sigma_1^1$-hard}.
\end{corollary}

\subsection{Constructing $\Psi$}
\label{sec-two-counters-formula}

For the rest of this section, we fix the following non-deterministic one-counter Minsky machines:
\begin{eqnarray*}
    \M_1 & \equiv & 1:\Ins_1^1;\cdots m: \Ins_m^1;\\ 
    \M_2 & \equiv & 1:\Ins_1^2;\cdots m: \Ins_m^2;
 \end{eqnarray*}
Furthermore, we fix a partition $I = (I_1,I_2)$ of $\{1,\ldots,m\}$. 

Roughly speaking, the formula $\Psi$ is obtained by ``merging'' the formulae $\psi_1$ and $\psi_2$ of Section~\ref{sec-counter} constructed for $\M_1$ and $\M_2$ using the disjoint sets of atomic propositions $\A^1$ and $\A^2$. The main modification is in the subformulae $\LTrans$ of $\psi_1$ and $\psi_2$, where we need to adjust the way of selecting the propositions of $\Labels$ passed to the successors so that the operational semantics of $\M_1 \times_I \M_2$ is reflected properly. When constructing $\psi_1$ and $\psi_2$, we assume the constants $q$, $\vec{\kappa}$, and $\gamma$ satisfying the same constraints as in Section~\ref{sec-counter} (these constants are used in both $\psi_1$ and $\psi_2$).

The set of atomic propositions used in $\Psi$ is $\A^1 \cup \A^2$. For every formula $\textit{Form}$ constructed in Section~\ref{sec-counter} and $k \in \{1,2\}$, we use $\textit{\textbf{Form}}^k$ to denote the formula obtained from $\textit{Form}$ by replacing all propositions of $\A$ with $\A^k$. For example, $t \models \pmb{\Ex{a,r_i}}^1$ if $t$ satisfies both $a^1$ and $r_i^1$, and no other proposition of $\A^1$. The formula does not say anything about 
the validity of the propositions of $\A^2$ in~$t$.

When we need to define new formulae over $\A^1$ and $\A^2$ with the same structure up to the upper indexes of atomic propositions, we write a definition of $\textit{Form}^k$ parameterized by the~$k$. For example, by stipulating $\textit{Form}^k \equiv h^k \vee \pmb{\Ex{a,r_i}}^k$, we simultaneously define $\textit{Form}^1 \equiv h^1 \vee \pmb{\Ex{a,r_i}}^1$ and $\textit{Form}^2 \equiv h^2 \vee \pmb{\Ex{a,r_i}}^2$.
Furthermore, for $k \in \{1,2\}$, the other element of $\{1,2\}$ is denoted by $k'$, i.e., $\{k,k'\} = \{1,2\}$. 

We put
\[
     \Psi \ \equiv \ \bInit^1 \wedge \bInit^2 \wedge \opg_{=1} (\Invariant \wedge \LPass)   
\]
where
\begin{eqnarray*}
   \LPass & \equiv & \ \left(\bigvee_{j=1}^m (\ell^1_j \wedge \ell^2_j)\right) \ \Rightarrow \ \opx_{>0}\left( \bigvee_{j=1}^{m} (\ell^1_j \wedge \ell^2_j)\right)\\[1ex]
    \Invariant  & \equiv & (\bFin^1 \vee \Transient^1 \vee \bFree^1)\\
    & \wedge & (\bFin^2 \vee \Transient^2 \vee \bFree^2)
\end{eqnarray*}
Intuitively, if $t \models \ell_j^1 \wedge \ell_j^2$, then the instructions $\Ins_j^1$ 
and $\Ins_j^2$ are ``simulated in parallel'' in $t$. The formula $\LPass$ enforces 
the existence of at least one immediate successor $t'$ of $t$ where the successor instructions  $\Ins_{j'}^1$ and $\Ins_{j'}^2$ are again simulated jointly. 
For all $k \in \{1,2\}$, we put
\begin{eqnarray*}
    \Transient^k &  \equiv & \bTrans^k \vee \bCTrans^k \vee \LTrans^k
\end{eqnarray*}
where
\begin{eqnarray*}
    \LTrans^k &  \equiv & \bigg(\bigvee_{j=1}^m (\ell_j^k \wedge  \ell_j^{k'}) \bigg) \Rightarrow \Simulate^k\\
    & \wedge & \bigg(\bigvee_{j=1}^m (\ell_j^k  \wedge   \neg \ell_j^{k'}) \bigg) \Rightarrow \Abandon^k
\end{eqnarray*}

As we shall see, for every state $t$ reachable from a state satisfying $\Psi$ we have that $t \models \ell_j^k$ for at most one $\ell_j \in \Labels$. Intuitively, the formula $\LTrans^k$ says that when $t \models \ell_j^k \wedge \ell_j^{k'}$, we continue with simulating the instruction $\Ins_j^k$. If $t$ satisfies a proposition $\ell_j^k \in \Labels^k$ but not the matching proposition $\ell_j^{k'} \in \Labels^{k'}$, we ``abandon'' the simulation. In the latter case, the immediate successors of $t$ must still have an appropriate structure so that no formula is ``spoilt'' in the immediate predecessor of~$t$. This is enforced by the formula $\Abandon^k$. More precisely, we put
\begin{eqnarray*}
\Abandon^k & \equiv & (\bZero^k \Rightarrow \OZer^k)
             \wedge  (\neg\bZero^k \Rightarrow \OPos^k)\,.
\end{eqnarray*}
The formulae $\OZer^k$ and $\OPos^k$ are defined as follows:
\begin{eqnarray*}
    \OZer^k & \equiv & \bigvee_{i \in \{0,\ldots,4\}}  ( r_i^k \wedge \OZsuc^k_i)\\  
    \OPos^k & \equiv & 
    \bigvee_{i \in \{0,\ldots,4\}} ( r_i^k \wedge \OPsuc_i^k \wedge \bInterval^k \wedge \bEq_i^k )
\end{eqnarray*}
where
\begin{eqnarray*}  
    \OZsuc^k_i & \equiv & \opx_{=1} \big(\pmb{\Ex{h,a,S(r_i)}}^k \vee \pmb{\Ex{h,b,S^2(r_i)}}^k 
            \vee \pmb{\Ex{h,c,S^2(r_i)}}^k\\
    && \hspace*{5ex}  \vee\ \pmb{\Ex{h,c,S^2(r_i),e}}^k\big) \ \wedge \ 
    \opx_{=1{-}q} \pmb{\Ex{h,c,S^2(r_i),e}}^k \\
    \OPsuc^k_i & \equiv & \opx_{=1} \big(\pmb{\Ex{a,S(r_i)}}^k \vee \pmb{\Ex{h,b,S^2(r_i)}}^k 
    \vee \pmb{\Ex{h,c,S^2(r_i)}}^k\\
&& \hspace*{5ex}  \vee\ \pmb{\Ex{h,c,S^2(r_i),e}}^k\big) \ \wedge \ 
\opx_{=1{-}q} \pmb{\Ex{h,c,S^2(r_i),e}}^k
\end{eqnarray*}

The formulae $\Simulate^k$, where $k \in \{1,2\}$, enforce the simulation of one computational step of $\M_1 \times_I \M_2$.
We put
\begin{eqnarray*}   
    \Simulate^k & \equiv &\bigvee_{i \in \{0,\ldots,4\}} \bigvee_{\ell \in \Labels} \bigvee_{x \in \{a,b\}} 
    \  \left(\pmb{\Ex{x,r_i,\ell}}^k \wedge \NStep^k_{i,\ell}\right)
\end{eqnarray*}
where the formula $\NStep^k_{i,\ell}$ is constructed as follows. Let $\Ins_j^k$ be the instruction associated with $\ell$, i.e., $\ell = \ell_j$. If $j \in I_k$, then
\[
    \NStep^k_{i,\ell} \ \equiv \bStep^k_{i,\ell}    
\]

Now let $j \not\in I_k$. For every $L \subseteq \Labels$ where $L$ has one or two elements, let $\bStep^k_{i,\ell}[L]$ be the formula obtained from $\bStep^k_{i,\ell}$ by substituting every occurrence of every set of target labels with $L^k = \{p^k \mid p \in L\}$. Note that this substitution affects only the ``big disjunction'' in the subformulae $\bZsuc_{i,\ell}^k$, $\bPsuc_{i,\ell}^k$, $\bIZsuc_{i,\ell}^k$, and $\bIPsuc_{i,\ell}^k$. Now, we distinguish two cases (recall that $k' \neq k$ is ``the other index'' of $\{1,2\}$).
\smallskip

(A) $\Ins_j^{k'} \equiv \textit{if } c{=}0 \textit{ then goto } L \textit{ else dec } c; \textit{ goto } L'$. Then,
\begin{eqnarray*}
    \NStep_{i,\ell}^k & \equiv &  \left(\bZero^{k'} \Rightarrow \bStep^k_{i,\ell}[L]\right)
     \wedge \ \left(\neg \bZero^{k'} \Rightarrow \bStep^k_{i,\ell}[L']\right)\\
\end{eqnarray*}

(B) $\Ins_j^{k'} \equiv \textit{inc } c; \textit{ goto } L$. Then,
\begin{eqnarray*}
    \NStep_{i,\ell}^k & \equiv &  \bStep^k_{i,\ell}[L]
\end{eqnarray*}
This completes the construction of $\Psi$.

\subsection{A Proof of Theorem~\ref{thm-two-counter}}





Theorem~\ref{thm-two-counter} is proven by reusing the arguments used in the proof of Theorem~\ref{thm-one-counter}  with some modifications and extensions. Let 
$\M_1 \times_I \M_2$ be a synchronized product of two non-deterministic one-counter Minsky machines with $m$~instructions, and let $\Psi$ be the formula constructed for $\M_1 \times_I \M_2$ in Section~\ref{sec-two-counters-formula}. The two claims of Theorem~\ref{thm-two-counter} are proven separately in the following subsections.

\subsubsection{A Proof of Theorem~\ref{thm-two-counter}~(A)}

Let $M =(S,P,v)$ be a Markov chain such that $s \models \Psi$ for some $s \in S$. For every $k \in \{1,2\}$, let $T^k$ be the set of all $t \in S$ reachable from $s$ such that $t \not\models h^k$. It is easy to verify that the conditions of Theorem~\ref{thm-area} are satisfied for both $T^1$ and $T^2$, where  $\vec{v}^1[t]$ and $\vec{v}^2[t]$ play the role of $\vec{v}[t]$, respectively. 

We show that $s$ simulates $\M_1 \times_I \M_2$. Clearly, $s$ represents the configuration $(1,0,0)$ because $s \models \bInit^1 \wedge \bInit^2$. Let $t$ be a state reachable from $s$ such that $t$ represents a configuration $(j,n_1,n_2)$. Let $n_1',n_2' \in \N$ be the (unique) counter values in a successor configuration of $(j,n_1,n_2)$, and let $L$ be the set of all $j'$ such that $(j,n_1,n_2) \leadsto (j',n_1',n_2')$. 

Observe that  $t \models \Simulate^1 \wedge \Simulate^2$ and hence $t \models \NStep^1_{i,\ell_j} \wedge \NStep^2_{i',\ell_j}$ for some $i,i' \in \{0,\ldots,4\}$. Let $t'$ be an immediate successor of~$t$ such that $t' \models \ell^1_{j'} \wedge \ell^2_{j'}$ for some $j' \leq m$. Then $j' \in L$ by the definition of $\NStep_{i,\ell}^k$. Furthermore, at least one such $t'$ must exist because $t \models \LPass$. By using the arguments of Section~\ref{sec-prop-one-counter-unbounded}, we obtain that $\vec{v}^1[t'] = \sigma^{n'_1}(\vec{\kappa})$ and $\vec{v}^2[t'] = \sigma^{n'_2}(\vec{\kappa})$. Hence, $t'$ represents a successor configuration of $(j,n_1,n_2)$. Also observe that if $t''$ is an immediate successor of $t$ satisfying the formula $\bigwedge_{j=1}^m \neg(\ell_j^1 \wedge \ell_j^2)$, then all states reachable from $t''$ also satisfy this formula.

\subsubsection{A Proof of Theorem~\ref{thm-two-counter}~(B)}

Let $\omega = D_0,D_1,\ldots$ be a computation of $\M_1 \times_I \M_2$. If $\omega$ is periodic, we fix $\alpha,\beta$ such that $\alpha < \beta$ and the computations $D_{\alpha-1},D_{\alpha},D_{\alpha+1},\ldots$ and $D_{\beta-1},D_{\beta},D_{\beta+1},\ldots$ are the same. If $\omega$ is not periodic, then $\beta = \infty$ and $\alpha$ is undefined. 

For every configuration $D_i = (j,n_1,n_2)$ of $\omega$, let $C_i^1 = (j,n_1)$ and $C_i^2 = (j,n_2)$ be the corresponding configurations of $\M_1$ and $\M_2$. Furthermore, we define the infinite sequences $\omega^1 = C_0^1,C_1^1,C_2^1,\ldots$ and $\omega^2 = C_0^2,C_1^2,C_2^2,\ldots$. Note that $\omega_1$ and $\omega_2$ are \emph{not} necessarily computations of $\M_1$ and $\M_2$. However, for all $k \in \{1,2\}$ and $i \in \N$, we have that if $C_i^k = (j,n)$ and $C_{i+1}^k = (j',n')$, then there is $j'' \leq m$ such that $(j,n) \mapsto (j'',n')$ is a computational step of $\M_k$. In other words, $n'$ is obtained from $n$ by executing $\Ins_j$ in $\M_k$.


We construct a Markov chain $M = (S,P,v)$ and a state $s \in S$ such that $s \models \Psi$ and $s$ covers $\omega$. If $\beta < \infty$, then $S$ is finite.

The states of $S$ are tuples of the form $t = [\iota,\calL_1,\calL_2,n_1,n_2]$ where $\iota \in \{i \in \N \mid 0 \leq i < \beta\} \cup \{\star\}$, $\calL_1,\calL_2 \subseteq \A$, and $n_1,n_2 \in \N \cup \{\star\}$. The set $v(t)$ of atomic propositions satisfied in $t$ is $\{p^1 \mid p \in \calL_1\} \cup \{p^2 \mid p \in \calL_2\}$. 

Each state $t = [\iota,\calL_1,\calL_2,n_1,n_2]$ of $S$ determines two \emph{projections} $t_1 = [\iota,\calL_1,n_1]$ and $t_2 = [\iota,\calL_2,n_2]$. Conversely, for all $u_1 = [\iota,\calL_1,n_1]$ and $u_2 = [\iota,\calL_2,n_2]$ where $\iota,\calL_1,\calL_2,n_1,n_2$ satisfy the conditions of the previous paragraph we define a tuple $u_1 \uplus u_2 = [\iota,\calL_1,\calL_2,n_1,n_2]$.

The Markov chain $M$ is the least Markov chain $M'$ such that $s = [0,\{a,r_0,\ell_1\},\{a,r_0,\ell_1\},0,0]$ is a state of $M'$, and if $t$ is a state of $M'$, then $M'$ contains all immediate successors of $t$ defined by the rules given below. The rules are designed so that for every $t \in S$, the immediate successors of the projections $t_1$ and $t_2$ are defined by the rules of Section~\ref{sec-proof-prop-one-counter-bounded} where $\M_1$ and $\M_2$ are used as the underlying one-counter Minsky machines, the infinite sequences $\omega_1$ and $\omega_2$ play the role of the fixed computation $\omega$, and the constants $\alpha$ and $\beta$ refer to the constants fixed above. Recall that $\omega$ is used only in Rule~I of Section~\ref{sec-proof-prop-one-counter-bounded}, and this rule makes a clear sense also for $\omega^1$ and $\omega^2$. Also recall that the set $\Succ(t_k)$ of all immediate successors of $t_k$ (where $k \in \{1,2\}$) satisfies precisely one of the following conditions:
\begin{itemize}
    \item $\Succ(t_k) = \{t_k\}$. This happens iff $t_k \models h$.
    \item $\Succ(t_k)$ contains precisely three states satisfying the propositions $a$, $b$, and $c$, respectively, such that the last state does not satisfy~$d \vee e$. These states are denoted by $\ta_k$, $\tb_k$, and $\tc_k$, respectively.
    \item $\Succ(t_k)$ has precisely four states; apart of $\ta_k$, $\tb_k$, and $\tc_k$, there is also the fourth state satisfying $d$ or $e$. This fourth state is denoted by $\te_k$.
\end{itemize}

\noindent
Now we can define the closure rules.
\smallskip

\noindent
\textbf{Rule A.} If $t = [k,\{x,r_i,\ell_j\},\{x',r_{i'},\ell_j\},n_1,n_2]$ where $x,x' \in \{a,b\}$, $i,i' \in \{0,\ldots,4\}$, then the immediate successors of $t$ are defined as follows. 
Let $t_1' \in \Succ(t_1)$ and $t_2' \in \Succ(t_2)$ be the unique states satisfying some proposition of $\Labels$. Observe that $t_1' = [k',\calL_1,n_1]$ and $t_2' = [k',\calL_2,n_2]$ where $\calL_1$ and $\calL_2$ contain the \emph{same} proposition of $\Labels$ (this follows immediately by inspecting Rule~I of Section~\ref{sec-proof-prop-one-counter-bounded} and the definitions of $\omega^1$ and $\omega^2$).  We put 
\[
  P(t,t_1' \uplus t_2') = \min\{P(t_1,t_1'),P(t_2,t_2')\}\,.
\]
Furthermore, 
\begin{itemize}
    \item if $P(t_1,t_1') > P(t_2,t_2')$, then 
    \[
       P(t,t_1' \uplus \tc_2) = P(t_1,t_1') - P(t_2,t_2');
    \] 
    \item if $P(t_2,t_2') > P(t_1,t_1')$, then 
    \[
        P(t,\tc_1 \uplus t_2') = P(t_2,t_2') - P(t_1,t_1');
    \] 
    \item for all $t_1'' \in \Succ(t_1)$ where $t_1'' \neq t_1'$ and $t_1'' \in \{\ta_1,\tb_1,\te_1\}$, 
    \[
        P(t,t_1'' \uplus \tc_2) = P(t_1,t_1'');
    \]
    \item for all $t_2'' \in \Succ(t_2)$ where $t_2'' \neq t_2'$ and $t_2'' \in \{\ta_2,\tb_2,\te_2\}$, 
    \[
        P(t,\tc_1 \uplus t_2'') = P(t_2,t_2'');
    \]
    \item finally, we put 
    \(
        P(t,\tc_1 \uplus \tc_2) = 1-s
    \)
    where $s$ is the sum of all $P(t,\cdot)$ defined above.
\end{itemize}
Note that $s < 1$ due to the constraints on $q$ and $\vec{\kappa}$ adopted in Section~\ref{sec-counter}.  
\smallskip

\noindent
\textbf{Rule~B.} If $t = [k,\calL_1,\calL_2,n_1,n_2]$ where $\calL_1 \cap \Labels \neq \emptyset$, \mbox{$\calL_2 \cap \Labels = \emptyset$}, and $h \not\in \calL_2$, then $\calL_1$ contains precisely one $r_i \in \{r_0,\ldots,r_4\}$, and we distinguish two subcases.
\smallskip

\noindent
\pmb{$n_1 = 0$}. Then
\begin{eqnarray*}
    P(t,[\star,\{h,a,S(r_i)\},\star] \uplus \tc_2) & = & \vec{\kappa}_1,\\ 
    P(t,[\star,\{h,b,S^2(r_i)\},\star] \uplus \tc_2) & = & \vec{\kappa}_2,\\
    P(t,[\star,\{h,c,S^2(r_i),e\},\star] \uplus \tc_2) & = & 1-q \,.
\end{eqnarray*}
Furthermore, for every $t_2'' \in \{\ta_2,\tb_2,\te_2\}$, we put
\begin{eqnarray*}
    P(t,[\star,\{h,c,S^2(r_i)\},\star] \uplus t_2'') & = & P(t_2,t_2'')\,.
\end{eqnarray*}
Finally, we put $P(t,[\star,\{h,c,S^2(r_i)\},\star] \uplus \tc_2)  =  1-s$
where $s$ is the sum of all $P(t,\cdot)$ defined above.
\smallskip

\noindent
\pmb{$n > 0$}. Then
\begin{eqnarray*}
    P(t,[\star,\{a,S(r_i)\},n{-}1] \uplus \tc_2) & = & \sigma^n(\vec{\kappa})_1,\\ 
    P(t,[\star,\{h,b,S^2(r_i)\},\star] \uplus \tc_2) & = & \sigma^n(\vec{\kappa})_2,\\
    P(t,[\star,\{h,c,S^2(r_i),e\},\star] \uplus \tc_2) & = & 1-q\,.
\end{eqnarray*}
Furthermore, for every $t_2'' \in \{\ta_2,\tb_2,\te_2\}$, 
\begin{eqnarray*}
    P(t,[\star,\{h,c,S^2(r_i)\},\star] \uplus t_2'') & = & P(t_2,t_2'')\,.
\end{eqnarray*}
Finally, we put $P(t,[\star,\{h,c,S^2(r_i)\},\star] \uplus \tc_2)  =  1-s$
where $s$ is the sum of all $P(t,\cdot)$ defined above.
\smallskip

\noindent
\textbf{Rule~C.} If $t = [k,\calL_1,\calL_2,n_1,n_2]$ where $\calL_2 \cap \Labels \neq \emptyset$, \mbox{$\calL_1 \cap \Labels = \emptyset$}, and $h \not\in \calL_1$, then the immediate successors of $t$ are defined similarly as in Rule~B (Rule~C is fully symmetric to Rule~B).
\smallskip

\noindent
\textbf{Rule~D.} If $t = [k,\calL_1,\calL_2,n_1,n_2]$ where $\calL_1 \cap \Labels = \emptyset$, $\calL_2 \cap \Labels = \emptyset$, $h \not\in \calL_1$, and $h \not\in \calL_2$, then the immediate successors or $t$ are defined as follows. For every 
$t_1'' \in \{\ta_1,\tb_1,\te_1\}$, 
\begin{eqnarray*}
    P(t, t_1'' \uplus \tc_2) & = & P(t_1,t_1'')\,.
\end{eqnarray*}
Similarly, for all $t_2'' \in \{\ta_2,\tb_2,\te_2\}$, 
\begin{eqnarray*}
    P(t, \tc_1 \uplus t_2'') & = & P(t_2,t_2'')\,.
\end{eqnarray*}
Finally, $P(t, \tc_1 \uplus \tc_2) = 1 -s$, where $s$ is the sum of all $P(t,\cdot)$ above.
\smallskip

\noindent
\textbf{Rule~E.} If $t = [k,\calL_1,\calL_2,n_1,\star]$ where $h \not\in \calL_1$ and $h\in \calL_2$, then
\(
     P(t, t_1' \uplus t_2) = P(t_1,t_1')   
\)
for all $t_1' \in \Succ(t_1)$.
\smallskip 

\noindent
\textbf{Rule~F.} If $t = [k,\calL_1,\calL_2,n_1,\star]$ where $h \in \calL_1$ and $h \not\in \calL_2$, then
\(
    P(t, t_1 \uplus t_2') = P(t_2,t_2')   
\)
for all $t_2' \in \Succ(t_2)$.
\smallskip 

\noindent
\textbf{Rule~G.} If $t = [k,\calL_1,\calL_2,n_1,\star]$ where $h \in \calL_1$ and $h \in \calL_2$, then $P(t,t) =1$.
\smallskip 

If the computation $\omega$ is periodic, then the constructed Markov chain~$M$ has finitely many states. A routine check (similar to the one performed in the proof of Theorem~\ref{thm-one-counter}~(B)) reveals that the state $s = [0,\{a,r_0,\ell_1\},\{a,r_0,\ell_1\},0,0]$ satisfies $\Psi$ and covers the computation $\omega$.

\section{Conclusions}
\label{sec-concl}

We have shown that the general/finite PCTL satisfiability problems are undecidable. 
Note that the formula $\Psi$ constructed in the proof of Theorem~\ref{thm-two-counter} is \emph{always} generally satisfiable, which implies that the finite satisfiability problem for PCTL is undecidable even for the subset of generally satisfiable PCTL formulae. Furthermore, the undecidability result remains valid even if the set of eligible finite models is restricted to \emph{tree-like} models, where the underlying graph of a Markov chain is a tree with self-loops on all leaves (this requires a slight modification of our construction).
Finally, let us note that the construction of Theorem~\ref{thm-two-counter} applies also to a universal Minsky machine. Hence, fixed parameterized PCTL formulae exist such that the general/finite satisfiability of their instances is undecidable. 


\newpage
\bibliographystyle{ACM-Reference-Format}
\bibliography{str-long,concur}


\begin{thebibliography}{23}


\ifx \showCODEN    \undefined \def \showCODEN     #1{\unskip}     \fi
\ifx \showDOI      \undefined \def \showDOI       #1{#1}\fi
\ifx \showISBNx    \undefined \def \showISBNx     #1{\unskip}     \fi
\ifx \showISBNxiii \undefined \def \showISBNxiii  #1{\unskip}     \fi
\ifx \showISSN     \undefined \def \showISSN      #1{\unskip}     \fi
\ifx \showLCCN     \undefined \def \showLCCN      #1{\unskip}     \fi
\ifx \shownote     \undefined \def \shownote      #1{#1}          \fi
\ifx \showarticletitle \undefined \def \showarticletitle #1{#1}   \fi
\ifx \showURL      \undefined \def \showURL       {\relax}        \fi
\providecommand\bibfield[2]{#2}
\providecommand\bibinfo[2]{#2}
\providecommand\natexlab[1]{#1}
\providecommand\showeprint[2][]{arXiv:#2}

\bibitem[Baier and Katoen(2008)]%
        {BK:book}
\bibfield{author}{\bibinfo{person}{C. Baier} {and} \bibinfo{person}{J.-P.{} Katoen}.} \bibinfo{year}{2008}\natexlab{}.
\newblock \bibinfo{booktitle}{\emph{Principles of Model Checking}}.
\newblock \bibinfo{publisher}{The MIT Press}.
\newblock


\bibitem[Baier and Kwiatkowska(1998)]%
        {BK:PCTL-fairness}
\bibfield{author}{\bibinfo{person}{C. Baier} {and} \bibinfo{person}{M. Kwiatkowska}.} \bibinfo{year}{1998}\natexlab{}.
\newblock \showarticletitle{Model checking for a probabilistic branching time logic with fairness}.
\newblock \bibinfo{journal}{\emph{Distributed Computing}} \bibinfo{volume}{11}, \bibinfo{number}{3} (\bibinfo{year}{1998}), \bibinfo{pages}{125--155}.
\newblock


\bibitem[Bertrand et~al\mbox{.}(2012)]%
        {BFS:bounded-PCTL}
\bibfield{author}{\bibinfo{person}{N. Bertrand}, \bibinfo{person}{J. Fearnley}, {and} \bibinfo{person}{S. Schewe}.} \bibinfo{year}{2012}\natexlab{}.
\newblock \showarticletitle{Bounded Satisfiability for {PCTL}}. In \bibinfo{booktitle}{\emph{Proceedings of CSL 2012}} \emph{(\bibinfo{series}{Leibniz International Proceedings in Informatics}, Vol.~\bibinfo{volume}{16})}. \bibinfo{publisher}{Schloss Dagstuhl--Leibniz-Zentrum f{\"{u}}r Informatik}, \bibinfo{pages}{92--106}.
\newblock


\bibitem[Bianco and de~Alfaro(1995)]%
        {BA:MDP-PCTL}
\bibfield{author}{\bibinfo{person}{A. Bianco} {and} \bibinfo{person}{L. de Alfaro}.} \bibinfo{year}{1995}\natexlab{}.
\newblock \showarticletitle{Model Checking of Probabilistic and Nondeterministic Systems}. In \bibinfo{booktitle}{\emph{Proceedings of {FST\&TCS'95}}} \emph{(\bibinfo{series}{Lecture Notes in Computer Science}, Vol.~\bibinfo{volume}{1026})}. \bibinfo{publisher}{Springer}, \bibinfo{pages}{499--513}.
\newblock


\bibitem[Billingsley(1995)]%
        {Billingsley:book}
\bibfield{author}{\bibinfo{person}{P. Billingsley}.} \bibinfo{year}{1995}\natexlab{}.
\newblock \bibinfo{booktitle}{\emph{Probability and Measure}}.
\newblock \bibinfo{publisher}{Wiley}.
\newblock


\bibitem[Br\'{a}zdil et~al\mbox{.}(2006)]%
        {BBFK:Games-PCTL-objectives}
\bibfield{author}{\bibinfo{person}{T. Br\'{a}zdil}, \bibinfo{person}{V. Bro{\v{z}}ek}, \bibinfo{person}{V. Forejt}, {and} \bibinfo{person}{A. Ku{\v{c}}era}.} \bibinfo{year}{2006}\natexlab{}.
\newblock \showarticletitle{Stochastic Games with Branching-Time Winning Objectives}. In \bibinfo{booktitle}{\emph{Proceedings of LICS 2006}}. \bibinfo{publisher}{IEEE Computer Society Press}, \bibinfo{pages}{349--358}.
\newblock


\bibitem[Br{\'{a}}zdil et~al\mbox{.}(2008)]%
        {BFKK:satisfiability}
\bibfield{author}{\bibinfo{person}{T. Br{\'{a}}zdil}, \bibinfo{person}{V. Forejt}, \bibinfo{person}{J. K{\v{r}}et{\'{\i}}nsk{\'{y}}}, {and} \bibinfo{person}{A. Ku{\v{c}}era}.} \bibinfo{year}{2008}\natexlab{}.
\newblock \showarticletitle{The Satisfiability Problem for Probabilistic {CTL}}. In \bibinfo{booktitle}{\emph{Proceedings of LICS 2008}}. \bibinfo{publisher}{IEEE Computer Society Press}, \bibinfo{pages}{391--402}.
\newblock


\bibitem[Br\'{a}zdil et~al\mbox{.}(2008)]%
        {BFK:MDP-PECTL-objectives}
\bibfield{author}{\bibinfo{person}{T. Br\'{a}zdil}, \bibinfo{person}{V. Forejt}, {and} \bibinfo{person}{A. Ku{\v{c}}era}.} \bibinfo{year}{2008}\natexlab{}.
\newblock \showarticletitle{Controller Synthesis and Verification for {Markov} Decision Processes with Qualitative Branching Time Objectives}. In \bibinfo{booktitle}{\emph{Proceedings of ICALP 2008, Part II}} \emph{(\bibinfo{series}{Lecture Notes in Computer Science}, Vol.~\bibinfo{volume}{5126})}. \bibinfo{publisher}{Springer}, \bibinfo{pages}{148--159}.
\newblock


\bibitem[Br{\'{a}}zdil et~al\mbox{.}(2005)]%
        {BKS:pPDA-temporal}
\bibfield{author}{\bibinfo{person}{T. Br{\'{a}}zdil}, \bibinfo{person}{A. Ku{\v{c}}era}, {and} \bibinfo{person}{O. Stra{\v{z}}ovsk{\'{y}}}.} \bibinfo{year}{2005}\natexlab{}.
\newblock \showarticletitle{On the Decidability of Temporal Properties of Probabilistic Pushdown Automata}. In \bibinfo{booktitle}{\emph{Proceedings of STACS 2005}} \emph{(\bibinfo{series}{Lecture Notes in Computer Science}, Vol.~\bibinfo{volume}{3404})}. \bibinfo{publisher}{Springer}, \bibinfo{pages}{145--157}.
\newblock


\bibitem[Chakraborty and Katoen(2016)]%
        {CHK:PCTL-simple}
\bibfield{author}{\bibinfo{person}{S.{} Chakraborty} {and} \bibinfo{person}{J.P.{} Katoen}.} \bibinfo{year}{2016}\natexlab{}.
\newblock \showarticletitle{On the Satisfiability of Some Simple Probabilistic Logics}. In \bibinfo{booktitle}{\emph{Proceedings of LICS 2016}}. \bibinfo{pages}{56--65}.
\newblock


\bibitem[Chodil and Ku{\v{c}}era(2024)]%
        {ChK:PCTL-quatitative-fragments-JCSS}
\bibfield{author}{\bibinfo{person}{M. Chodil} {and} \bibinfo{person}{A. Ku{\v{c}}era}.} \bibinfo{year}{2024}\natexlab{}.
\newblock \showarticletitle{The Satisfiability Problem for a Quantitative Fragment of {PCTL}}.
\newblock \bibinfo{journal}{\emph{J. Comput. System Sci.}}  \bibinfo{volume}{139} (\bibinfo{year}{2024}), \bibinfo{pages}{103478}.
\newblock


\bibitem[Emerson(1991)]%
        {Emerson:temp-logic-handbook}
\bibfield{author}{\bibinfo{person}{E.A.{} Emerson}.} \bibinfo{year}{1991}\natexlab{}.
\newblock \showarticletitle{Temporal and Modal Logic}.
\newblock \bibinfo{journal}{\emph{Handbook of Theoretical Computer Science}}  \bibinfo{volume}{B} (\bibinfo{year}{1991}), \bibinfo{pages}{995--1072}.
\newblock


\bibitem[Esparza et~al\mbox{.}(2006)]%
        {EKM:prob-PDA-PCTL-LMCS}
\bibfield{author}{\bibinfo{person}{J. Esparza}, \bibinfo{person}{A. Ku{\v{c}}era}, {and} \bibinfo{person}{R. Mayr}.} \bibinfo{year}{2006}\natexlab{}.
\newblock \showarticletitle{Model-Checking Probabilistic Pushdown Automata}.
\newblock \bibinfo{journal}{\emph{Logical Methods in Computer Science}} \bibinfo{volume}{2}, \bibinfo{number}{1:2} (\bibinfo{year}{2006}), \bibinfo{pages}{1--31}.
\newblock


\bibitem[Etessami and Yannakakis(2012)]%
        {EY:RMC-LTL-complexity-TCL}
\bibfield{author}{\bibinfo{person}{K. Etessami} {and} \bibinfo{person}{M. Yannakakis}.} \bibinfo{year}{2012}\natexlab{}.
\newblock \showarticletitle{Model Checking of Recursive Probabilistic Systems}.
\newblock \bibinfo{journal}{\emph{{ACM} Transactions on Computational Logic}}  \bibinfo{volume}{13} (\bibinfo{year}{2012}).
\newblock
Issue 2.


\bibitem[Hansson and Jonsson(1994)]%
        {HJ:logic-time-probability-FAC}
\bibfield{author}{\bibinfo{person}{H. Hansson} {and} \bibinfo{person}{B. Jonsson}.} \bibinfo{year}{1994}\natexlab{}.
\newblock \showarticletitle{A Logic for Reasoning about Time and Reliability}.
\newblock \bibinfo{journal}{\emph{Formal Aspects of Computing}}  \bibinfo{volume}{6} (\bibinfo{year}{1994}), \bibinfo{pages}{512--535}.
\newblock


\bibitem[Harel(1986)]%
        {Harel:Infinite-trees-JACM}
\bibfield{author}{\bibinfo{person}{D. Harel}.} \bibinfo{year}{1986}\natexlab{}.
\newblock \showarticletitle{Effective Transformations on Infinite Trees with Applications to High Undecidability Dominoes, and Fairness}.
\newblock \bibinfo{journal}{\emph{Journal of the Association for Computing Machinery}} \bibinfo{volume}{33}, \bibinfo{number}{1} (\bibinfo{year}{1986}).
\newblock


\bibitem[Hart and Sharir(1984)]%
        {HS:Prob-temp-logic}
\bibfield{author}{\bibinfo{person}{S. Hart} {and} \bibinfo{person}{M. Sharir}.} \bibinfo{year}{1984}\natexlab{}.
\newblock \showarticletitle{Probabilistic Temporal Logic for Finite and Bounded Models}. In \bibinfo{booktitle}{\emph{Proceedings of POPL'84}}. \bibinfo{publisher}{ACM Press}, \bibinfo{pages}{1--13}.
\newblock


\bibitem[Huth and Kwiatkowska(1997)]%
        {HK:quantitative-mu-calculus-LICS}
\bibfield{author}{\bibinfo{person}{M. Huth} {and} \bibinfo{person}{M.Z. Kwiatkowska}.} \bibinfo{year}{1997}\natexlab{}.
\newblock \showarticletitle{Quantitative Analysis and Model Checking}. In \bibinfo{booktitle}{\emph{Proceedings of LICS'97}}. \bibinfo{publisher}{IEEE Computer Society Press}, \bibinfo{pages}{111--122}.
\newblock


\bibitem[Kraus and Lehmann(1983)]%
        {KL:qPCTL-satisfiability}
\bibfield{author}{\bibinfo{person}{S. Kraus} {and} \bibinfo{person}{D.J.{} Lehmann}.} \bibinfo{year}{1983}\natexlab{}.
\newblock \showarticletitle{Decision Procedures for Time and Chance (Extended Abstract)}. In \bibinfo{booktitle}{\emph{Proceedings of FOCS'83}}. \bibinfo{publisher}{IEEE Computer Society Press}, \bibinfo{pages}{202--209}.
\newblock


\bibitem[K{\v{r}}et{\'{\i}}nsk{\'{y}} and Rotar(2018)]%
        {KR:PCTL-unbounded}
\bibfield{author}{\bibinfo{person}{J. K{\v{r}}et{\'{\i}}nsk{\'{y}}} {and} \bibinfo{person}{A. Rotar}.} \bibinfo{year}{2018}\natexlab{}.
\newblock \showarticletitle{The Satisfiability Problem for Unbounded Fragments of Probabilistic {CTL}}. In \bibinfo{booktitle}{\emph{Proceedings of CONCUR 2018}} \emph{(\bibinfo{series}{Leibniz International Proceedings in Informatics}, Vol.~\bibinfo{volume}{118})}. \bibinfo{publisher}{Schloss Dagstuhl--Leibniz-Zentrum f{\"{u}}r Informatik}, \bibinfo{pages}{32:1--32:16}.
\newblock


\bibitem[Kuzmin et~al\mbox{.}(2010)]%
        {KSCh:Minsky-boundedness-PCS}
\bibfield{author}{\bibinfo{person}{E.V.{} Kuzmin}, \bibinfo{person}{V.A.{} Sokolov}, {and} \bibinfo{person}{D.Y.{} Chalyy}.} \bibinfo{year}{2010}\natexlab{}.
\newblock \showarticletitle{Boundedness problems for {Minsky} counter machines}.
\newblock \bibinfo{journal}{\emph{Programming and Computer Software}} \bibinfo{volume}{36}, \bibinfo{number}{1} (\bibinfo{year}{2010}), \bibinfo{pages}{3--10}.
\newblock


\bibitem[Lehmann and Shelah(1982)]%
        {LS:time-chance-IC}
\bibfield{author}{\bibinfo{person}{D. Lehmann} {and} \bibinfo{person}{S. Shelah}.} \bibinfo{year}{1982}\natexlab{}.
\newblock \showarticletitle{Reasoning with Time and Chance}.
\newblock \bibinfo{journal}{\emph{Information and Control}}  \bibinfo{volume}{53} (\bibinfo{year}{1982}), \bibinfo{pages}{165--198}.
\newblock


\bibitem[Webster(1994)]%
        {Webster:book}
\bibfield{author}{\bibinfo{person}{R. Webster}.} \bibinfo{year}{1994}\natexlab{}.
\newblock \bibinfo{booktitle}{\emph{Convexity}}.
\newblock \bibinfo{publisher}{Oxford University Press}.
\newblock


\end{thebibliography}


\clearpage
\appendix

\begin{center}
  \huge\bf Appendix
\end{center}

Here, we provide the proofs omitted in the main body of the paper due to space constraints.

\tausigma*
\begin{proof}
Recall that 
\begin{eqnarray*}
    I_q & = & \left(\frac{1-\sqrt{4q-3}}{2},  \frac{1+\sqrt{4q-3}}{2}\right)\\
    W   & = & I_q \times [0,\infty)\\
    \tau(\vec{v}) & = & \left( \frac{q{-}1{+}\vec{v}_1}{\vec{v}_1}, \frac{\vec{v}_2}{\vec{v}_1} \right)\\
    \sigma(\vec{v}) & = & \left( \frac{1{-}q}{1{-}\vec{v}_1}, \frac{\vec{v}_2(1{-}q)}{1{-}\vec{v}_1} \right) 
\end{eqnarray*}
\smallskip

\textit{Item~(a).}
Let $\vec{v} \in W$. We show that $\tau(\vec{v}) \in W$. Observe
\begin{eqnarray*}
   \tau(\vec{v})_1 & = & 1 - \frac{1-q}{\vec{v}_1}\\
     & < & 1 - \frac{2(1-q)}{1-\sqrt{4q-3}}\\
     & = & 1- \frac{(2-2q)}{1-\sqrt{4q-3}} \cdot \frac{1+\sqrt{4q-3}}{1+\sqrt{4q-3}}\\
     & = & 1- \frac{(2-2q)(1+\sqrt{4q-3})}{4-4q}\\
     & = & \frac{1+\sqrt{4q-3}}{2}
\end{eqnarray*}
Similarly, we obtain $\tau(\vec{v})_1 > \frac{1-\sqrt{4q-3}}{2}$, and hence $\tau(\vec{v})_1 \in I_q$. Since $\tau(\vec{v})_2 = \vec{v}_2/\vec{v}_1 > 0$, we have that $\tau(\vec{v}) \in W$ as required.

Now we show that $\sigma(\vec{v}) \in W$. Observe
\begin{eqnarray*}
    \sigma(\vec{v})_1 & = & \frac{1-q}{1-\vec{v}_1}\\
      & < & \frac{1-q}{1-\frac{1+\sqrt{4q-3}}{2}}\\[2ex]
      & = & \frac{2(1-q)}{1- \sqrt{4q-3}} \cdot \frac{1+\sqrt{4q-3}}{1+\sqrt{4q-3}}\\[2ex]
      & = & \frac{1+\sqrt{4q-3}}{2}
\end{eqnarray*}
Similarly, we obtain $\sigma(\vec{v})_1 >  \frac{1-\sqrt{4q-3}}{2}$. Since $\sigma(\vec{v})_2 = \frac{\vec{v}_2(1{-}q)}{1{-}\vec{v}_1} > 0$, we have that $\sigma(\vec{v}) \in W$.
\smallskip

\textit{Item~(b).}
Let $\vec{v} \in W$. Observe that $\tau(\vec{v})_1 > \vec{v}_1$ iff $(q{-}1{+}\vec{v}_1)/\vec{v}_1 > \vec{v}_1$ iff $\vec{v}_1^2 - \vec{v}_1 - q +1 < 0$ iff $\vec{v}_1 \in I_q$. This explains our choice of $I_q$. Furthermore, $\tau(\vec{v})_2 = \vec{v}_2/\vec{v}_1 \geq \vec{v}_2$; if $\vec{v}_2 > 0$, then $\tau(\vec{v})_2 > \vec{v}_2$.
\smallskip

\textit{Item~(c).} Let $\vec{v} \in W$. Observe
\begin{eqnarray*}
    \slope(\vec{u},\tau(\vec{v})) & = & \frac{\tau(\vec{v})_2}{\tau(\vec{v})_1 - \vec{v}_1} = \frac{\vec{v}_2}{q-1+\vec{v}_1(1 - \vec{v}_1)}
 \end{eqnarray*}
 Similarly,
 \begin{eqnarray*}
    \slope(\tau(\vec{v}),\tau^2(\vec{v})) & = &  \frac{\tau(\vec{v})_2(1-\tau(\vec{v})_1)}{q-1+\tau(\vec{v})_1( 1 - \tau(\vec{v}_1))}\\[1ex]
    & = & \frac{\frac{\vec{v}_2}{\vec{v}_1}\left(1-\frac{q-1+\vec{v}_1}{\vec{v}_1}\right)}%
               {q-1 + \frac{q-1+\vec{v}_1}{\vec{v}_1} \left(1 - \frac{q-1+\vec{v}_1}{\vec{v}_1}\right)}\\[1ex]
    & = & \frac{\frac{\vec{v}_2}{\vec{v}_1}\left(\frac{1-q}{\vec{v}_1}\right)}%
    {q-1 + \frac{q-1+\vec{v}_1}{\vec{v}_1} \left(\frac{1-q}{\vec{v}_1}\right)}\\[1ex]  
    & = & \frac{\vec{v}_2}{q-1+\vec{v}_1(1-\vec{v}_1)}       
 \end{eqnarray*}
Hence, $\slope(\vec{u},\tau(\vec{v})) = \slope(\tau(\vec{v}),\tau^2(\vec{v}))$.
\smallskip

\textit{Item~(d).} Realize that 
\begin{eqnarray*}    
    \slope(\vec{u},\tau(\vec{u})) & = & \frac{y(1-\vec{v}_1)}{q-1+\vec{v}_1(1 - \vec{v}_1)}\\[1ex]
    \slope(\vec{v},\tau(\vec{v})) & = & \frac{\vec{v}_2(1-\vec{v}_1)}{q-1+\vec{v}_1(1 - \vec{v}_1)}
\end{eqnarray*}
SInce $0 \leq y < \vec{v}_2$, we have that $\slope(\vec{u},\tau(\vec{u})) < \slope(\vec{v},\tau(\vec{v}))$.

\smallskip

\textit{Item~(e).} It is trivial to verify that $\sigma(\tau(\vec{v})) = \tau(\sigma(\vec{v}) = \vec{v}$ for every $\vec{v} \in W$.
\end{proof}

\gproperties*
\begin{proof}
Let $\vec{v} \in W$. All claims follow directly from Lemma~\ref{lem-tausigma}. More concretely, for every $\vec{u} \in \Points(\vec{v})$, we have that both $\vec{u}$ and $L(\vec{u})$ are \emph{faces} of the convex set $\Area(\vec{v})$ (see, e.g., Section~2.6 in~\cite{Webster:book}), and the claims~(A) and~(B) are just instances of the defining property of a face. 
\end{proof}

\tauconvex*
\begin{proof}
It is easy to verify that for all $\vec{x},\vec{y} \in W$ and all $\lambda \in (0,1]$ we have that
\[
    \tau(\lambda\vec{x} + (1{-}\lambda)\vec{y}) \ = \ 
    \lambda'\tau(\vec{x}) + (1{-}\lambda') \tau(\vec{y})
\]
where 
\[
    \lambda' = \frac{\lambda \vec{x}_1}{\lambda\vec{x}_1 + (1{-}\lambda) \vec{y}_1} 
\]
Observe that $\lambda' \in (0,1]$. The lemma follows by putting $\vec{x} = \vec{w}$, $\vec{y} = \tau(\vec{w})$ and choosing $\lambda$ so that $\vec{u} = \lambda\vec{w} + (1{-}\lambda)\tau(\vec{w})$. 
\end{proof}  

\outlineseg*
\begin{proof}
First, we show that 
\begin{eqnarray}
    \lim_{k \to \infty} \sigma^k(\vec{\kappa})_1 & = & \frac{1-\sqrt{4q-3}}{2} \label{lim-2}
\end{eqnarray}
Let 
\[
    J_q =    \left[\frac{1-\sqrt{4q-3}}{2},  \vec{\kappa}_1\right] 
\]
By Lemma~\ref{lem-tausigma}, the infinite sequence  $\vec{\kappa}_1, \sigma(\vec{\kappa})_1, \sigma^2(\vec{\kappa})_1,\ldots$ is decreasing and bounded from below by $(1-\sqrt{4q-3})/2$. Consequently, the sequence has a limit~$\alpha \in J_q$, and hence it is also a Cauchy sequence, i.e.,
\[
    \lim_{k \to \infty} \sigma^{k+1}(\vec{\kappa})_1 - \sigma^{k}(\vec{\kappa})_1 = 0 \,.
\] 
Consider the function $f : J_q \to \R$ where 
\[
    f(x) \ = \ \frac{1-q-x(1-x)}{1-x}
\]
Observe that $f$ is non-negative and continuous. Furthermore, $f(x) = 0$ iff $x =  (1 - \sqrt{4q-3})/2$. Observe that for every $k \in \N$, we have that
\begin{eqnarray*}
    \sigma^{k+1}(\vec{\kappa})_1 - \sigma^{k}(\vec{\kappa})_1 \ = \ 
        f(\sigma^{k}(\vec{\kappa})_1) \,.
\end{eqnarray*}
Hence, 
\[
    0 \ = \ \lim_{k \to \infty} \sigma^{k+1}(\vec{\kappa})_1 - \sigma^{k}(\vec{\kappa})_1  =  \lim_{k \to \infty} f(\sigma^{k}(\vec{\kappa})_1) \ = \ f(\alpha)
\]
which implies $\alpha = (1 - \sqrt{4q-3})/2$.

Now let $\vec{v} \in W \smallsetminus \Area(\vec{\kappa})$ where $\vec{v}_1 \leq \vec{\kappa}_1$. By~\eqref{lim-2}, there exist $k \in \N$ such that $\sigma^k(\vec{\kappa})_1 \leq \vec{v}_1 < \sigma^{k-1}(\vec{\kappa})_1$. We put $\vec{u}_1 = \sigma^k(\vec{\kappa})_1$ and choose $\vec{u}_2$ so that 
\[
    \slope(\vec{u},\tau(\vec{u})) = \slope(\vec{u},\vec{v})\,.
\]
Hence, we require that
\[
     \frac{\vec{u}_2(1-\vec{u}_1)}{q-1+\vec{u}_1(1-\vec{u}_1)}    =
     \frac{\vec{v}_2 - \vec{u}_2}{\vec{v}_1 - \vec{u}_1}
\]
From this, we obtain
\[
    \vec{u}_2 = \frac{\vec{v}_2(q-1+\vec{u}_1(1-\vec{u}_1))}{\vec{v}_1 - \vec{u}_1 + q - 1 + \vec{u}_1(1-\vec{u}_1)}    
\]
and the proof is finished.
\end{proof}    

\Minskyproduct*
\begin{proof}
    Let $\M \equiv 1:\Ins_1;\cdots m: \Ins_m;$ be a non-deterministic two-counter Minsky machine. We start by transforming $\M$ into another two-counter Minsky machine $\widehat{\M}$ with $3m$ instructions constructed as follows:
    \begin{itemize}
        \item The first $m$ instructions of $\widehat{\M}$ are the same as the instructions of $\M$.
        \item For every $j \in \{1,\ldots,m\}$, the machine $\widehat{M}$ contains the following labeled instructions:
        \begin{itemize}
            \item $\ \ m{+}j: \textit{if } c_1{=}0 \textit{ then goto } \{1\} \textit{ else dec } c_1; \textit{ goto } \{j\}$
            \item $2m{+}j: \textit{if } c_2{=}0 \textit{ then goto } \{1\} \textit{ else dec } c_2; \textit{ goto } \{j\}$
        \end{itemize}
        As we shall see, the target labels in the \textit{then} branches are insignificant and can be chosen arbitrarily.
    \end{itemize}
    Observe that $\widehat{\M}$ has the same set of computations as $\M$, because the newly added instructions are not reachable from the initial configuration. Furthermore, if $\M$ is deterministic, then $\widehat{\M}$ is also deterministic.
    
    Let $L_1,L_2 \subseteq \{1,\ldots,3m\}$ be the sets of labels of all instructions of $\widehat{\M}$ operating on $c_1$ and $c_2$, respectively. Now, we construct two one-counter Minsky machines $\M_1,\M_2$ whose synchronized product simulates $\widehat{\M}$.
    
    Both $\M_1$ and $\M_2$ have $6m$ instructions. 
    For notation convenience, the labels of $\M_1,\M_2$ are written as pairs $(\ell,0)$, $(\ell,+)$, where $\ell \in \{1,\ldots, 3m\}$.  
    
    For every $\ell \in \{1,\ldots, 3m\}$, the instructions labeled by $(\ell,0)$ and $(\ell,+)$ are constructed as follows: 
    Let $\Ins_\ell$ be the instruction of $\widehat{\M}$ with label $\ell$. If $\ell \in L_1$, then  
    \begin{itemize}
        \item $\M_1$ contains the instruction $(\ell,0): \overline{\Ins}$, where  $\overline{\Ins}$ is obtained from $\Ins_\ell$ as follows:
        \begin{itemize}
            \item $c_1$ is replaced with $c$;
            \item each set of target labels $L$ occurring in $\Ins_\ell$ is replaced with $\overline{L}$ obtained from $L$ by replacing every $u \in L$ with either $(u,+)$ or $(2m{+}u,0)$, depending on whether \mbox{$u \in L_1$} or $u \in L_2$, respectively.
        \end{itemize}
        \item $\M_2$ contains the instruction $(\ell,0): \textit{inc } c; \textit{goto } \{1\}$
        \item $\M_1$ contains the instruction $(\ell,+): \overline{\Ins}$, where  $\overline{\Ins}$ is obtained from $\Ins_\ell$ as follows:
        \begin{itemize}
            \item $c_1$ is replaced with $c$;
            \item each set of target labels $L$ occurring in $\Ins_\ell$ is replaced with $\overline{L}$ obtained from $L$ by replacing every $u \in L$ with $(u,0)$.
        \end{itemize}
        \item $\M_2$ contains the instruction\\ $(\ell,+): \textit{if } c{=}0 \textit{ then goto } \{1\} \textit{ else dec } c; \textit{ goto } \{1\}$ 
    \end{itemize}
    If $\ell \in L_2$, then 
    \begin{itemize}
        \item $\M_2$ contains the instruction $(\ell,0): \overline{\Ins}$, where  $\overline{\Ins}$ is obtained from $\Ins_\ell$ as follows:
        \begin{itemize}
            \item $c_2$ is replaced with $c$;
            \item each set of target labels $L$ occurring in $\Ins_\ell$ is replaced with $\overline{L}$ obtained from $L$ by replacing every $u \in L$ with either $(u,+)$ or $(m{+}u,0)$, depending on whether \mbox{$u \in L_2$} or $u \in L_1$, respectively.
        \end{itemize}
        \item $\M_1$ contains the instruction $(\ell,0): \textit{inc } c; \textit{goto } \{1\}$
        \item $\M_2$ contains the instruction $(\ell,+): \overline{\Ins}$, where  $\overline{\Ins}$ is obtained from $\Ins_\ell$ as follows:
        \begin{itemize}
            \item $c_2$ is replaced with $c$;
            \item each set of target labels $L$ occurring in $\Ins_\ell$ is replaced with $\overline{L}$ obtained from $L$ by replacing every $u \in L$ with $(u,0)$.
        \end{itemize}
        \item $\M_1$ contains the instruction\\ $(\ell,+): \textit{if } c{=}0 \textit{ then goto } \{1\} \textit{ else dec } c; \textit{ goto } \{1\}$ 
    \end{itemize}
    Furthermore, we put $I = (I_1,I_2)$, where $I_1$ is the set of all $(\ell,0),(\ell,+)$ such that $\ell \in L_1$, and $I_2$ contains the other labels. The computation of $\M_1 \times_I \M_2$ starts by executing the instructions with label $(1,0)$.
    
    Intuitively, $\M_1 \times_I \M_2$ simulates  $\widehat{\M}$ where the instructions on $c_1$ and $c_2$ are performed by $\M_1$ and $\M_2$, respectively.  
    As long as $\M_1$ performs instructions on $c_1$, $\M_2$ keeps incrementing/decrementing $c_2$ alternately (the flags $+$ and $0$ in the label indicate whether the ``inactive'' counter should be decremented or incremented; note that we only decrement the inactive counter when it was incremented before, and hence its value is certainly positive). When an instruction operating on $c_2$ is reached, the control is passed to $\M_2$. Before executing the instruction on $c_2$, $\M_2$ possibly decrements $c_2$ to restore its value. 
    It is easy to check that $\M$ is deterministic and bounded iff $\M_1 \times_I \M_2$ is deterministic and bounded, and $\M$ has a recurrent computation iff $\M_1 \times_I \M_2$ has a recurrent computation.

\end{proof}

\end{document}